%% file: main.tex
\documentclass[11pt,hyphens]{article}
\usepackage[utf8]{inputenc}

\usepackage{fullpage}

\usepackage{amsmath}
\usepackage{amsthm}
\usepackage{amssymb}
\usepackage{natbib}
\usepackage{color}
\usepackage{mathrsfs}
\usepackage{comment}
\usepackage[dvipsnames]{xcolor}
\usepackage{hyperref}
\hypersetup{
    colorlinks,
    linkcolor=BrickRed,
    citecolor=NavyBlue,
    urlcolor=Cerulean,
    linktoc=page
}
\usepackage{tikz}
\usepackage{dsfont}
\usepackage{bbm}
\usepackage{soul}
\usepackage{braket}

\usepackage{caption}
\usepackage{subcaption}

\input{macros}

\newcommand{\degminus}{\mathrm{adeg}^{-}}
\newlang{\exone}{E1}

\title{QMA Lower Bounds for Batch Verification via Approximate Degree}
\author{%
Mark Bun\thanks{Mark Bun was supported in part by NSF award CNS-2046425 and a Sloan Research Fellowship.} \\
Boston University \\
\texttt{mbun@bu.edu}
\and
Mandar Juvekar\thanks{Mandar Juvekar was supported in part by NSF award CNS-2046425 and an IBM PhD Fellowship.} \\
Boston University \\
\texttt{mandarj@bu.edu}
\and
Samuel King\thanks{Samuel King was supported in part by NSF CAREER award CCF-2338730.} \\
Georgetown University \\
\texttt{sik29@georgetown.edu}
}

\begin{document}

\maketitle

\begin{abstract}
    	We study batch verification in QMA query and communication complexity, where the goal is to understand how the resources needed to verify $m$ copies of a Boolean function $f$ depend on $m$. We give a general technique for proving lower bounds on the witness-query tradeoff needed to batch verify a function $f$ in terms of its approximate degree. Applying this technique to an explicit family of DNF formulas $f$, we show that attempting to save even a constant factor on the witness length of the baseline approach to batch verifying $f$ necessitates a large polynomial increase in the query cost. We also obtain new lower bounds on the QMA query complexity of read-once CNF formulas and on the surjectivity and $k$-element distinctness functions.
        Our lower bounds also lift to give communication analogs of these results.
\end{abstract}

\tableofcontents
\clearpage

\section{Introduction}
\input{sections/intro}

\section{Preliminaries} \label{sec:prelims}
\input{sections/prelims}

\section{SBQP Lower Bounds from High Approximate Degree}
\input{sections/sbqp-lower-bound}

\section{\QMAdt{} Lower Bounds}
\input{sections/qma-dt-batch-ver}

\section{Witness/Communication Tradeoffs for \QMAcc{} Batch Verification}
\input{sections/qma-cc-batch-ver}

\section*{Acknowledgments}
We would like to thank Nadezhda Voronova for her insightful comments on a draft of this paper, Sasha Golovnev for invaluable discussions during this work, and Robin Kothari and Justin Thaler for stimulating our interest in these problems.
We would also like to thank Robin Kothari for spotting an error in a previous manuscript.

\bibliographystyle{alpha}
\bibliography{mainbib}

\end{document}

%% file: macros.tex
\usepackage{complexity}

\newtheorem{theorem}{Theorem}[section]

\newtheorem{lemma}[theorem]{Lemma}

\newtheorem{corollary}[theorem]{Corollary}

\newtheorem{remk}[theorem]{Remark}

\theoremstyle{definition}
\newtheorem{definition}[theorem]{Definition}

\renewcommand{\R}{{\mathbb R}}
\newcommand{\N}{{\mathbb{N}}}

\newcommand{\zo}{\{0,1\}}
\newcommand{\pmone}{\{-1,1\}}

\let\E\relax
\DeclareMathOperator*{\E}{\mathbb{E}}

\newcommand{\eps}{\varepsilon}
\newcommand{\one}{\mathbf{1}}

\newcommand{\numberthis}{\addtocounter{equation}{1}\tag{\theequation}}

\newcommand{\abs}[1]{\left\lvert #1 \right\rvert}
\newcommand{\norm}[1]{\left\lVert #1 \right\rVert}
\newcommand{\ang}[1]{\left\langle #1 \right\rangle}

\newclass{\dt}{dt}
\newclass{\SBQP}{SBQP}
\newclass{\coQMA}{coQMA}

\newlang{\exhalf}{ExactHalf}
\newlang{\UOR}{UOR}
\newlang{\SURJ}{SURJ}
\newlang{\DIST}{DIST}
\newlang{\ED}{ED}

\newcommand{\QMAdt}{\textup{QMA}$^{\textup{dt}}$}
\newcommand{\QMAcc}{\textup{QMA}$^{\textup{cc}}$}

\newcommand{\sgn}{\operatorname{sgn}}

\newlang{\THR}{THR}
\newlang{\MAJ}{MAJ}
\newlang{\AND}{AND}
\newlang{\OR}{OR}
\newlang{\XOR}{XOR}

\newlang{\EQ}{EQ}
\newlang{\GT}{GT}
\newlang{\PrOR}{PrOR}

\newcommand{\apdeg}{\mathrm{adeg}}

%% file: sections/intro.tex
A powerful but untrusted server wants to convince a computationally constrained client that it has correctly performed a large list of computations that were delegated to it.
It can do so by proving to the client the veracity of $m$ statements $x_1, \ldots, x_m$.
As communication is costly, the proof must be as short as possible.
If the statements are instances of an NP (or more generally for this work, QMA) language, a baseline strategy is to send as a proof witnesses $w_1, \ldots, w_m$ for all of the statements---the client can then verify each statement individually.
Both the communication and the client's computational cost in this strategy are $m$ times the cost for a single statement.
For what statements is it possible to do better?
What tradeoffs are achievable between communication and computation?

This \emph{batch verification} problem has generated exciting research in complexity theory and cryptography. Highly communication-efficient interactive proofs are known for batch verifying UP statements (i.e., NP statements with unique witnesses) \cite{rei-rot-rot:delegating-computation, rei-rot-rot:batch-verification-up, rot-rot:batch-verification-proximity}, while a connection between batch verification and statistical witness indistinguishability~\cite{bit-etal:batch-statistically-hiding} suggests extending such a result to all of NP may be impossible. Zero-knowledge batch verification protocols are known for problems admitting non-interactive statistical zero-knowledge proofs~\cite{kas-etal:batch-zero-knowledge, KaslasiRV21, MuNRV24, KeretRV24}. Under cryptographic assumptions, computationally sound noninteractive batch arguments are known for all of NP~\cite{KalaiPY19, ChoudhuriJJ21, ChoudhuriJJ21b, DevadasGKV22, WatersW22, KalaiLVW23}; recent work has even constructed succinct interactive (classical) batch arguments for QMA~\cite{GoyalJM25}.
In cryptography, the study of batch verification of NP statements has led to new constructions of succinct non-interactive arguments (SNARGs)~\cite{ChoudhuriJJ21b, KalaiVZ21} and non-interactive zero-knowledge arguments (NIZKs)~\cite{ChampionW23} for more expressive classes and from weaker assumptions.
Finally,~\cite{NassarR25} recently initiated a study of quantum witness indistinguishable proofs and showed that they are implied by quantum batch proofs.

In this work, we study the batch verification problem for quantum Merlin-Arthur (QMA) proofs in the query and communication models. These models abstract computation to settings where techniques are available for proving strong lower bounds, yet which shed both conceptual and technical light on the corresponding Turing machine model. In the QMA \emph{query} or \emph{black-box} model, hereafter denoted by \QMAdt{}, a quantum verifier (the resource-bounded mortal ``Arthur'') wishes to evaluate a known function $f : \pmone^n \to \pmone$ on an input $x \in \pmone^n$.%
\footnote{The notation \QMAdt{} alludes to yet another view of query algorithms as generalizations of decision trees. Here and throughout the paper we identify $-1$ with ``true'' and 1 with ``false.''}
He receives a witness consisting of $w$ qubits from a prover (the all-powerful wizard ``Merlin''), and is allowed to make $q$ queries to $x$ in quantum superposition. The protocol is correct if (1) whenever $f(x) = -1$ (i.e., $x$ is a ``true'' input), there exists a witness that causes the verifier to accept with probability at least $2/3$, and (2) whenever $f(x) = 1$ (i.e., $x$ is a ``false'' input), every witness causes the verifier to reject with probability at least $2/3$. The cost of such a protocol is the sum $w + q$ of the witness length and the number of verifier queries, though we are also interested in understanding the tradeoff between these parameters. The QMA query complexity of a Boolean function $f$, denoted $\QMA^\dt(f)$, is the least cost of a protocol computing $f$.

The batch verification problem for QMA query complexity can thus be stated as follows: under what conditions on $f$ does the function $(\AND_m \circ f^m)(x_1, \dots, x_m) := f(x_1) \land \dots \land f(x_m)$ have low \QMAdt{} complexity?
But what should ``low'' mean here?
In the quantum setting, Grover's search algorithm gives a way to ``batch compute'' all functions with bounded-error algorithms: when $f$ has a bounded-error quantum query protocol (i.e., one that does not make use of a witness) with query cost $q$, Grover search implies that $\AND_m \circ f^m$ has a bounded-error quantum query protocol with cost only $O(\sqrt{m} \cdot q)$ \cite{hoy-mos-dew:quantum-search}.
This methodology extends to \QMAdt{} protocols.
If $f$ admits a \QMAdt{} protocol with witness length $w_f$ and query complexity $q$, one can compute $\AND_m \circ f^m$ by taking as witness all $m$ individual witnesses (for a total witness length of $w = m \cdot w_f$) and using Grover search to verify all $m$ copies using $O(\sqrt{m} \cdot q)$ queries. Our batch verification problem thus becomes: for what functions $f$ can one beat this baseline protocol for $\AND_m \circ f^m$ that uses Grover search to verify all $m$ individual witnesses?

Our first main result establishes a strong limitation on batching even when $f$ is highly structured and admits an efficient NP query protocol.

\begin{theorem}[Informal] \label{thm:informal-dnf}
For every $\delta > 0$ there is an explicitly given polynomial-size, constant-width DNF $f : \pmone^n \to \pmone$ such that every \QMAdt{} protocol for $\AND_m \circ f^m$, $m = n^{O(1)}$, with witness length $w \le c m \log n$ for a suitable constant $c \in (0,1)$ must make at least $\widetilde{\Omega}(n^{1-\delta} \sqrt{m/w})$ queries.
\end{theorem}

One can interpret this result as follows. Since the function $f$ is computed by a polynomial-size, constant-width DNF, it admits an NP query protocol with witness length $O(\log n)$ and $O(1)$ queries. The baseline strategy for verifying $\AND_m \circ f^m$ thus entails a proof with witness length $O(m \log n)$ and uses $O(\sqrt{m})$ queries. Theorem~\ref{thm:informal-dnf} shows that if one seeks to beat this baseline using a protocol with witness length just a constant factor smaller than this $O(m \log n)$, then the requisite query complexity jumps to $\Omega(n^{1-\delta})$, which is $\Omega(m)$ whenever $m = O(n^{1-\delta})$.

We prove Theorem~\ref{thm:informal-dnf} by proving a more general connection between the \QMAdt{} complexity of batch verifying $f$ with the \emph{approximate degree} of $f$.
The $\eps$-approximate degree of $f$, denoted $\apdeg_\eps(f)$, is the minimum degree of a real-valued polynomial that approximates $f$ point-wise to additive error $\eps$.
Approximate degree is an extensively studied measure, and has many connections to randomized and quantum query and communication complexity (see, e.g., \cite{bun-tha:approximate-degree}). In particular, Beals et al.~\cite{bea-etal:quantum-polynomials} showed that the acceptance probability of every $q$-query quantum algorithm can be written as a degree-$2q$ polynomial in its input, and hence the $\eps$-error quantum query complexity of a function $f$ is always at least half of its $\eps$-approximate degree.

We show that lower bounds on the $\eps$-approximate degree of $f$ for various settings of $\eps$ yield witness-query tradeoffs for \QMAdt{} batch verifying $f$.
\begin{theorem}[Informal] \label{thm:informal-generic}
    Let $f : \pmone^n \to \pmone$ be a function.
    \begin{enumerate}
        \item If $\apdeg_{1 - 1/m}(f) > d$ then every \QMAdt{} protocol for  $\AND_m \circ f^m$ with witness length $w = O(m)$ requires $\Omega(d \sqrt{m/w})$ queries.
        \item If $\apdeg_{1/3}(f) > d$
        then every \QMAdt{} protocol for $\AND_m \circ f^m$ with witness length at most $m$ requires $\Omega(d)$ queries.
    \end{enumerate}
\end{theorem}
As we discuss below, previous work derived \QMAdt{} lower bounds from more stringent variants of approximate degree lower bounds.
However, to our knowledge, ours is the first to draw a direct connection between \QMAdt{} complexity and the standard notion of approximate degree.

As an illustrative example, let $f$ be the complement of the element distinctness function on $n$ bits (defined in Section~\ref{sec:poly-preliminaries}), which is known to have approximate degree $\tilde{\Omega}(n^{2/3})$~\cite{aar-shi:element-distinctness,amb:element-distinctness-range}.
At the same time, $f$ admits a \QMAdt{} protocol with witness length $O(\log n)$ and query complexity $O(\sqrt{\log n})$: the witness consists of the location of a collision, and the verifier uses Grover search to check the collision.
So the na\"ive batching strategy achieves witness length $O(m \log n)$ and query complexity $O(\sqrt{m \log n})$.
Item 2 of Theorem~\ref{thm:informal-generic} shows that any protocol improving on this witness length by even an $O(\log n)$ factor must have its query complexity jump to $\widetilde{\Omega}(n^{2/3})$, which is significantly larger for any $m \ll n^{4/3}$. Item~1 of the theorem gives more fine-grained tradeoffs when one can obtain lower bounds on the approximate degree of $f$ to \emph{large} error.
Similar tradeoffs apply to other functions that admit efficient \QMAdt{} protocols such as DNFs and the complement of the surjectivity function.

Beyond the intrinsic interest of the batch verification problem, Theorem~\ref{thm:informal-generic} also serves as a valuable technical tool for understanding the \QMAdt{} complexity of specific functions that have been central to the study of quantum query complexity.
Sherstov and Thaler~\cite{she-tha:vanishing-error} used \QMAdt{} lower bounds on functions of the form $\AND_m \circ f^m$ to give the sharpest known lower bounds for the influential element distinctness and permutation testing problems. For technical reasons elaborated on below, their method falls short of applying to closely related functions, such as the two-level AND-OR tree, the $k$-element distinctness function, and the surjectivity function. Our framework for proving batch verification lower bounds bypasses the limitations from prior work, allowing us to prove the following lower bounds for these functions.

\begin{theorem}\label{thm:specific}
    Let $\SURJ_n$ denote the surjectivity function on $n$ bits  and $k\ED_n$ the $k$-element distinctness function on $n$ bits.
    \begin{enumerate}
        \item $\QMA^\dt(\AND_{n^{1/3}} \circ \OR_{n^{2/3}}^{n^{1/3}}) = \Omega(n^{1/3})$.
        \item $\QMA^\dt(\SURJ_n) = \widetilde{\Omega}(n^{3/7})$.
        \item $\QMA^\dt(k\ED_n) = \widetilde{\Omega}(n^{\frac{3k-1}{7k-1}})$.
    \end{enumerate}
\end{theorem}

The best lower bounds that followed straightforwardly from prior work were $\QMA^\dt(\AND_{n^{1/3}} \circ \OR_{n^{2/3}}^{n^{1/3}}) = \Omega(n^{1/6})$, $\QMA^\dt(\SURJ_n) = \widetilde{\Omega}(n^{1/4})$ and $\QMA^\dt(k\ED_n) \geq \QMA^\dt(\ED_n) = \widetilde{\Omega}(n^{2/5})$.

\subsection{Technical Overview}

Our main results follow from a general result (Corollaries~\ref{c:large-error-to-sbqp} and \ref{c:const-error-to-sbqp}, which formalize Theorem~\ref{thm:informal-generic})
showing that $\AND_m \circ f^m$ has high \QMAdt{} complexity whenever $f$ has high \emph{approximate degree}.

Suppose there is a QMA protocol for a function $F$ using witness length $w$ and $q$ queries. The in-place amplification technique of Marriott and Watrous~\cite{mar-wat:quantum-arthur} gives rise to a protocol for $F$ with improved error $2^{-8w}$ using witness length $w$ and $O(qw)$ queries. Replacing the witness in this protocol with the maximally mixed state yields a witness-free quantum query protocol for $F$, and hence an approximating polynomial $p$ with the following properties:
\begin{gather*}
    F(x) = -1 \implies  -2^{8w}\le p(x) \le -1, \\
    F(x) = 1 \implies 1 - 2^{-8w} \le p(x) \le 1.
\end{gather*}
That is, $p$ very tightly approximates $F$ on its false inputs, but (roughly) subject to having the right sign behavior, can become exponentially large on true inputs. Thus, in order to prove a lower bound on the QMA query complexity of $F$, it suffices to rule out low-degree approximating polynomials $p$ of this form. Up to this point, this lower bound method, which is inspired by the complexity class containment $\QMA \subseteq \SBQP$, is well-established for proving QMA query complexity lower bounds~\cite{Aaronson12, aar-etal:approximate-counting, she-tha:vanishing-error, dal-etal:quantum-proofs}. We refer to such an approximating polynomial $p$ as an ``SBQP approximation'' to $F$.

Prior work proved degree lower bounds for SBQP approximations by relating it to other better-studied notions of polynomial approximation, including $\varepsilon$-approximate degree in the ``large error'' regime where $\varepsilon \approx 1 - 2^{-w}$, threshold degree, and vanishing error one-sided approximate degree. We depart from prior work by studying SBQP approximations more directly. We do this using the method of dual polynomials, whereby one proves approximability lower bounds by constructing dual solutions to a linear program capturing the approximation problem. We introduce a notion of ``relaxed SBQP duals'' to capture feasible solutions in this dual space that are relatively clean to construct and reason about. A relaxed SBQP dual for a function $F$ is a real-valued function $\Phi$ that is $\varepsilon$-correlated with $F$ for some $\varepsilon \ge 2^{-8w}$, orthogonal to all low degree polynomials, and which has ``almost one-sided error'' in the sense that the total weight it places on $\{x : F(x) = -1 \land \Phi(x) > 0\}$ is at most $2^{-8w}\varepsilon$. The proof of our general theorem proceeds by showing that whenever $f$ has high approximate degree, we can construct an appropriate relaxed SBQP dual for the function $\AND_m \circ f^m$. The connection outlined above, in turn, implies our \QMAdt{} tradeoffs.

To construct this relaxed SBQP dual, we appeal to the fact that the high approximate degrees of $\AND_m$ and $f$ are themselves witnessed by dual objects $\psi$ and $\varphi$, respectively.
Building on a construction of Sherstov~\cite{she:direct-product}, we show how to combine $\psi$ and $\varphi$ to produce a new dual witness to the high SBQP degree of the composed function $F = \AND_m \circ f^m$. Note that Sherstov's original construction was tailored to showing only that $F$ has high approximate degree. Since we need a stronger SBQP degree lower bound, our analysis needs to first take advantage of the fact that $\psi$ actually witnesses the high \emph{one-sided} approximate degree of $\AND_m$. That is, it rules out one-sided approximations which tightly approximate $\AND_m$ on false inputs, but may be arbitrarily unbounded on true inputs. This manifests in the dual witness $\psi$ as the property that it places zero weight on the set $\{x : \AND_m(x) = -1 \land \psi(x) > 0\}$. Second, our analysis needs to keep tighter track of how errors on the ``true'' side of $F$ accumulate to show that it has the required almost one-sided error.

With Theorem~\ref{thm:informal-generic} established, Theorem~\ref{thm:informal-dnf} follows by taking the $f$ to be an explicit, constant-width, polynomial-size DNF with approximate degree $\Omega(n^{1-\delta})$ as constructed by Sherstov~\cite{she:dnf-cnf}. We obtain the results in Theorem~\ref{thm:specific} by building on the ideas of Sherstov and Thaler~\cite{she-tha:vanishing-error}. For every parameter $m$, the functions $k\ED_n$ and $\SURJ_n$ contain, as subfunctions, $\AND_m \circ (k\ED_{n/m})^m$ and $\AND_m \circ \SURJ_{n/m}^m$, respectively. We can thus invoke our general theorem on these subfunctions for appropriate choices of $m$ that balance the final witness length and query costs in our lower bound.

The reason why Sherstov and Thaler could use this insight to prove lower bounds for element distinctness and permutation testing, but not for $k$-element distinctness or surjectivity, owes to a gap between the type of polynomial approximation lower bounds they could prove on functions of the form $\AND_m \circ f^m$ and the type they identified as yielding SBQP approximation lower bounds. Specifically, they showed a lower bound on the vanishing-error approximate degree $\apdeg_{\varepsilon}(\AND_m \circ f^m)$, but needed to invoke a stronger lower bound on the \emph{one-sided} vanishing-error approximate degree $\degminus_{\varepsilon}(\AND_m \circ f^m)$ to rule out SBQP approximations. For some functions with special structure, including $\AND$, element distinctness, and permutation testing, ordinary approximate degree and one-sided approximate degree fortuitously coincide, and hence vanishing-error approximate degree lower bounds imply \QMAdt{} lower bounds. But such a coincidence is not (yet) known for $k$-distinctness or surjectivity, and false for other simple functions such as the two-level AND-OR tree. By working directly with SBQP approximations via our relaxed SBQP duals, we can prove \QMAdt{} lower bounds for arbitrary composed functions $\AND_m \circ f^m$, regardless of whether ordinary and one-sided approximate degree agree.

\subsection{Lifting to QMA Communication Lower Bounds}

In the QMA \emph{communication} model, denoted by  \QMAcc{}, Alice holds an input $x \in \pmone^n$, Bob holds an input $y \in \pmone^n$, and they wish to compute a joint function $f : \pmone^n \times \pmone^n \to \pmone$ of their inputs. Bob receives a witness of $w$ qubits from the prover and the parties have unlimited access to shared entanglement. They wish to minimize the amount of communication between them (including the witness) required to compute $f(x, y)$. The success criterion, cost of a protocol, and \QMAcc{} complexity of a function $f$ are defined analogously to the query model.

Sherstov's pattern matrix method~\cite{she:pattern-matrix} shows how to generically ``lift'' an approximate degree lower bound for a function $f$ to a communication lower bound for a related function $G := f \circ g^n$, where $g$ is a certain constant-sized two-party gadget. By adapting the pattern matrix method to show how to lift SBQP degree lower bounds to QMA communication lower bounds, we establish QMA lower bounds for communication analogs of the batch verification problems we study. Our analysis builds closely on prior generalizations of the pattern matrix method lifting one-sided approximate degree to (multiparty) MA communication~\cite{gav-she:np-conp, she:dnf-cnf} and to QMA communication~\cite{kla:arthur-merlin-communication}, but the formulation of the statement that is most useful to us appears to be new. 

\subsection{Open Questions}
\begin{itemize}
\item Is the witness-query tradeoff given by Theorem~\ref{thm:informal-generic} tight? That is, can one find a function $f$ such that for every $w = O(m)$, $\AND_m \circ f^m$ admits a \QMAdt{} protocol with witness length $w$ and query complexity $O(\apdeg_{1-1/m}(f) \cdot \sqrt{m/w})$?

\item Our \QMAdt{} batch verification lower bounds depend on approximate degree lower bounds for the base function $f$. Recall that approximate degree is a lower bound on the quantum query complexity of $f$, but can be much smaller. Can one obtain similar batch verification lower bounds starting from arbitrary quantum query lower bounds on $f$?

\item How tight are the \QMAdt{} lower bounds in Theorem~\ref{thm:specific}?
Can our ideas be used to obtain improved \QMAdt{} lower bounds for element distinctness and permutation testing?
\end{itemize}

%% file: sections/prelims.tex
We let $\N = \{1, 2, \ldots\}$ denote the set of natural numbers and $\R$ denote the set of real numbers.
We write $[n]$ for the set $\{1, 2, \ldots, n\}$.

We identify the Boolean values ``true'' and ``false'' with $-1$ and $1$ respectively.
The functions $\OR_n : \pmone^n \to \pmone$ and $\AND_n : \pmone^n \to \pmone$ are the usual functions $\OR_n(x) = \bigvee_{i = 1}^n x_i$ and $\AND_n(x) = \bigwedge_{i = 1}^n x_i$. The Hamming weight of a string $x \in \pmone^n$, denoted $|x|$, is the number of indices $i$ such that $x_i = -1$.
A function $f$ is said to be computable by a \emph{disjunctive normal form (DNF)} formula of size $s$ and width $w$ if there exist $S_1, S_2, \ldots, S_s \subseteq [n]$ with cardinality $w$ such that $f(x) = \bigvee_{i \in [s]} \bigwedge_{j \in S_i} x_j$.
The conjunctions $\bigwedge_{j \in S_i} x_j$ are called \emph{terms}.

We use the following standard real-valued functions. We use $\exp_k$ to denote the base-$k$ exponential, i.e., $\exp_k(x) = k^x$. The base-2 logarithm is denoted by $\log(\cdot)$. For any predicate $P$, $\mathbbm{1}[P(x)]$ equals 1 if $P(x)$ is true and 0 otherwise.
The notation $\abs{x}$ represents the absolute value, Hamming weight, or cardinality of $x$ depending on if $x$ is a real number, binary string, or set respectively.

For real valued functions $f$ and $g$ with the same finite support $X$, we use the usual notation for the $\ell_1$ norm and inner product.
\begin{align*}
    \norm{f}_1 &= \sum_{x \in X} \abs{f(x)} \\
    \ang{f, g} &= \sum_{x \in X} f(x) g(x)
\end{align*}
We often refer to the inner product $\ang{f, g}$ as the \emph{correlation} between $f$ and $g$.
The function $f$ is \emph{balanced} if $\sum_{x \in X} f(x) = 0$.

We use superscripts to denote the direct product of a function with itself: for every $f : X \to Y$ and $m \in \N$, $f^m : X^m \to Y^m$ is the function $f^m(x_1, \ldots, x_m) = (f(x_1), \ldots, f(x_m))$.
The symbol $\circ$ is used for function composition, so that $(f \circ g)(x) = f(g(x))$.
Consequently, for every $g : X \to Y$ and $f : Y^m \to Z$, $f \circ g^m$ denotes the ``block composition'' of $f$ and $g$: $(f \circ g^m)(x_1, \ldots, x_m) = f(g(x_1), \ldots, g(x_m))$.

For any set $S$, we use $2^S$ to denote the set of all subsets of $S$.
For every $T \subseteq [n]$, we denote by $\chi_T : \pmone^n \to \pmone$ the Fourier character $\chi_T(x) = \prod_{i \in T} x_i$.
We recall that for every function $f : \pmone^n \to \R$ there exists a unique function $\widehat{f} : 2^{[n]} \to \R$ called the \emph{Fourier transform} of $f$ such that for all $x$,
\begin{equation*}
    f(x) = \sum_{T \subseteq [n]} \widehat{f}(T) \chi_T(x).
\end{equation*}
Furthermore, this unique function is given by
\begin{equation*}
    \widehat{f}(T) = 2^{-n} \sum_{x \in \pmone^n} f(x) \chi_T(x).
\end{equation*}

\subsection{Polynomial approximation} \label{sec:poly-preliminaries}

We say that a polynomial $p$ $\eps$-\emph{approximates} a function $f : \pmone^n \to \pmone$ if for all $x \in \pmone^n$, $\abs{f(x) - p(x)} \leq \eps$.
The $\eps$-\emph{approximate degree}, denoted $\apdeg_\eps(f)$, of $f$ is the minimum degree of a polynomial that $\eps$-approximates $f$.
A polynomial $p$ $\eps$-\emph{one-sided approximates} $f$ if the following conditions hold.
\begin{enumerate}
    \item For every $x \in f^{-1}(-1)$,  we have $p(x) \in (-\infty, -1+\eps]$.
    \item For every $x \in f^{-1}(1)$, we have $p(x) \in [1-\eps, 1+\eps]$.
\end{enumerate}
The $\eps$-\emph{one-sided approximate degree} of $f$,
denoted $\degminus_\eps(f)$, is the minimum degree of a polynomial that $\eps$-one-sided approximates $f$.
One can also define a complementary notion of one-sided approximation where the approximating polynomial must tightly approximate true inputs and is allowed to be arbitrarily positive on false inputs, though we do not consider that notion or its associated one-sided approximate degree measure in this paper.
For both one-sided and regular approximate degree, when $\eps$ is not specified, we assume that $\eps = 1/3$.

We recall the standard dual formulation of approximate degree.
\begin{theorem}[\cite{she:pattern-matrix, bun-tha:hardness-amplification}]
    \label{t:apdeg-dual}
    Let $f : \pmone^n \to \pmone$ be a Boolean function. Let $d \geq 1$ be an integer and $\eps > 0$ be a real number. Then $\apdeg_\eps(f) > d$ if and only if there exists a function $\psi : \pmone^n \to \R$ such that
    \begin{enumerate}
        \item $\ang{f, \psi} \geq \eps$,
        \item $\norm{\psi}_1 = 1$, and
        \item \label{item:adeg-phd} for every $T \subseteq [n]$ with $|T| \leq d$, $\ang{\psi, \chi_T} = 0$.
    \end{enumerate}
    Furthermore, $\degminus_\eps(f) > d$ if and only if there exists a function $\psi$ that satisfies the conditions above and has the additional property that for all $x$ such that $f(x) = -1$, $\psi(x) \leq 0$.
\end{theorem}
The function $\psi$ is called a \emph{dual polynomial} for the lower bound $\apdeg_\eps(f) > d$ or $\degminus_\eps(f) > d$ (depending on whether it has the additional property).
Condition~\ref{item:adeg-phd} is often referred to as $\psi$ having ``pure high degree at least $d$.''

Sherstov showed that there exist functions computable by polynomial-size, constant-width DNF formulas that have near-maximal approximate degree.

\begin{theorem}[\cite{she:dnf-cnf}] \label{t:cnf-high-odeg}
    For all constants $\delta > 0$ and $c \geq 1$, for every $n \in \N$, there is an explicit function $f_n : \pmone^n \to \pmone$ computable by a polynomial-size, constant-width DNF such that $\apdeg_{1 - 1/n^c}(f_n) = \Omega(n^{1-\delta})$.%
\end{theorem}
``Explicit'' here means that there is a uniform algorithm that, given $n$, prints out a description of $f_n$ in time polynomial in $n$.

The \emph{surjectivity} function $\SURJ_{N,R} : \pmone^{N \log R} \to \pmone$ interprets its input as a list of $N$ numbers in $[R]$ and evaluates to true if and only if every element of $[R]$ appears in the list at least once.
As a default we let $n = N \log R$ and $R = N/2$.
We write $\SURJ_n$ for $\SURJ_{N,R}$ with the default setting of parameters.
Bun, Kothari, and Thaler~\cite{bun-kot-tha:strikes-back} showed that the approximate degree of $\SURJ_n$ is $\widetilde{\Omega}(n^{3/4})$.

The \emph{element distinctness} problem $\ED_{N,R} : \pmone^{N \log R} \to \pmone$ interprets its input as a list of $N$ numbers in $[R]$ and evaluates to true if and only if no element appears more than once.
The canonical setting of parameters is $N = R$, as the function is trivially false when $N > R$ and it was shown by Ambainis~\cite{amb:element-distinctness-range} that the $\eps$-approximate degree of $\ED_{N,R}$ is the same for all $N \leq R$.
Aaronson and Shi~\cite{aar-shi:element-distinctness} showed that $\apdeg(\ED_n) = \widetilde{\Theta}(n^{2/3})$ when $R$ is large compared to $N$.
The aforementioned result of Ambainis extended their bounds to our setting of parameters.

The \emph{$k$-element distinctness} problem $k\ED_{N,R} : \pmone^{N \log R} \to \pmone$ is a generalization of element distinctness that evaluates to true if and only if no range element appears $k$ or more times.
We use the same canonical setting of parameters as for element distinctness: $k\ED_n = k\ED_{N,N}$ where $n = N \log N$.
Mande, Thaler, and Zhu~\cite{man-tha-zhu:k-distinctness} showed that $\apdeg(k\ED_n) = \widetilde{\Omega}(n^{3/4 - 1/(4k)})$.

Finally, for our lower bound on the \QMAdt{} complexity of the two-level AND-OR tree we will need the lower bound $\apdeg(\AND_m \circ \OR_n) = \Omega(\sqrt{mn})$, proved by Bun and Thaler~\cite{bun-tha:markov-bernstein} and Sherstov~\cite{she:approximating-and-or}.

\subsection{Quantum Merlin-Arthur protocols}

Let $f : \pmone^n \to \pmone$ be a function.
A \emph{QMA query protocol} for $f$ with witness length $w$, query complexity $q$, completeness error $\eps_1$, and soundness error $\eps_2$, is a quantum oracle algorithm $A$ such that
\begin{enumerate}
    \item For all $x \in f^{-1}(-1)$ there exists an $w$-qubit ``witness'' $\ket{\psi}$ such that $\Pr_A[A^x(\ket{\psi}) = -1] \geq 1 - \eps_1$;
    \item For all $x \in f^{-1}(1)$, for every $w$-qubit state $\ket{\psi}$, $\Pr_A[A^x(\ket{\psi}) = -1] \leq \eps_2$; and
    \item For all $x$ and $w$, $A^x(w)$ makes at most $q$ queries to the input $x$.
\end{enumerate}
The algorithm $A$ is often referred to as a \emph{verifier}.
If a protocol has completeness (resp.\ soundness) error 0, we say that it has perfect correctness (resp.\ soundness).
When the completeness and/or soundness errors are not specified, we set them to the default value $1/3$.
The QMA query complexity of $f$, denoted $\QMA^\dt(f)$, is the minimum $w + q$ such that there exists a QMA query protocol for $f$ with witness length $w$ and query complexity $q$.

QMA communication complexity is defined analogously.
A QMA communication protocol for a two-party function $f : X \times Y \to \pmone$ with witness length $w$, communication cost $c$, completeness error $\eps_1$, and soundness error $\eps_2$ is a two-party quantum communication protocol $(A, B)$ with access to shared entanglement such that for all inputs $(x, y)$ and $w$-qubit states $\ket{\psi}$, the execution $[A(x), B(y,\ket{\psi})]$ involves at most $c$ qubits of communication and satisfies completeness and soundness conditions analogous to the query case.
For a more formal definition of this setting we refer the reader to~\cite{raz-shp:quantum-proofs}.
The QMA communication complexity of $f$, denoted $\QMA^\cc(f)$ is the minimum $w + c$ such that there exists a QMA communication protocol for $f$ with witness length $w$ and communication cost $c$.

In models with classical witnesses (such as MA or QCMA), one can reduce completeness and soundness errors by repeatedly executing the verifier and outputting the majority output.
This does not work for protocols where the witness is a quantum state, since executing the verifier might alter the witness state.
Nevertheless, Marriott and Watrous showed that one can, in fact, reduce the completeness and soundness errors for QMA protocols.
The proof of this theorem can be adapted to achieve the same error reduction in the QMA communication setting.

\begin{theorem}[\cite{mar-wat:quantum-arthur}] \label{t:marriott-watrous}
    Let $f : \pmone^n \to \pmone$ be a function. Suppose there exists a $q(n)$-query quantum algorithm $Q$, numbers $a(n), b(n) \in [0,1]$, and $w(n), p(n) \leq \poly(n)$ such that
    \begin{enumerate}
        \item for all $x \in f^{-1}(-1)$ there exists a $w(n)$-qubit state $\ket{\psi}$ such that $\Pr_Q[Q^x(\ket{\psi})$ accepts$] \geq a(n)$,
        \item for all $x \in f^{-1}(1)$, for all $w(n)$-qubit states $\ket{\psi}$, $\Pr_Q[Q^x(\ket{\psi})$ accepts$] \leq b(n)$, and
        \item $a(n) - b(n) \geq 1/p(n)$.
    \end{enumerate}
    Then for every $r(n) \leq \poly(n)$ there exists a \QMAdt{} protocol for $f$ with completeness and soundness error $2^{-r(n)}$, witness length $w(n)$, and query complexity $O(q(n) \cdot r(n) \cdot p(n)^2)$.
\end{theorem}

\subsection{Linear algebra}

Our results on communication complexity require some standard concepts from matrix analysis which we recall here.
Unless otherwise specified, all matrices we consider are over the field of real numbers.
We use the notation $[A_{x,y}]_{x \in X, y \in Y}$ to refer to a $|X| \times |Y|$ matrix $A$ whose rows and columns are indexed by elements of $X$ and $Y$ respectively and whose entries are $A_{x,y}$.
For any two matrices $[A_{x,y}]_{x \in X, y \in Y}$ and $[B_{x,y}]_{x \in X, y \in Y}$, the inner product of $A$ and $B$ is defined as
\begin{equation*}
    \ang{A, B} = \sum_{x \in X, y \in Y} A_{x,y} B_{x,y}.
\end{equation*}
The singular values of an $m \times n$ matrix $A$ are denoted by $\sigma_1(A) \geq \sigma_2(A) \geq \cdots \geq \sigma_{\min\{m,n\}}(A) \geq 0$.
We recall four standard notions of matrix norm.
\begin{align*}
    \norm{A}_1 &= \sum_{x, y} |A_{x,y}| \tag{The $L_1$ norm} \\
    \norm{A} &= \max_i \sigma_i(A) \tag{The spectral norm} \\
    \norm{A}_\Sigma &= \sum_i \sigma_i(A) \tag{The trace norm} \\
    \norm{A}_F &= \sqrt{\sum_i \sigma_i(A)^2} \tag{The Frobenius norm}
\end{align*}
We will use the following well-known identities.
For all compatible real matrices $A$ and $B$,
\begin{align*}
    \norm{AB}_\Sigma &\leq \norm{A}_F \norm{B}_F, \text{ and} \\
    \ang{A, B} &\leq \norm{A}_\Sigma \norm{B}.
\end{align*}

%% file: sections/sbqp-lower-bound.tex
In this section we show that if a function $f$ has high large-error approximate degree, then $\AND_m \circ f^m$ has high SBQP cost.
We start by defining a notion of SBQP cost in terms of polynomial approximations and deriving a corresponding dual notion for it.

\begin{definition}
    A polynomial $(M,\eps)$-SBQP-approximates $f: \pmone^n \to \pmone$ if for all $x \in \pmone^n$,
    \begin{enumerate}
        \item if $f(x) = -1$ then $-M-1-\eps \leq p(x) \leq -1 + \eps$, and
        \item if $f(x) = 1$ then $1 - \eps \leq p(x) \leq 1 + \eps$.
    \end{enumerate}
\end{definition}

\begin{lemma} \label{l:sbqp-dual}
    For all $f : \pmone^n \to \pmone$ and $M^*, \eps, d \geq 0$, if there exists a function $\varphi : \pmone^n \to \R$ with the following properties, then there does not exist a polynomial of degree at most $d$ that $(M^*, \eps)$-SBQP-approximates $f$.
    \begin{enumerate}
        \item $\ang{\varphi, f} \geq 2\eps$,
        \item $\norm{\varphi}_1 \leq 1$,
        \item for all $S \subseteq [n]$ such that $|S| \leq d$, $\ang{\varphi, \chi_S} = 0$,
        \item $\sum_{x : f(x) = -1, \varphi(x) > 0} \varphi(x) \leq \eps/M^*$.
    \end{enumerate}
\end{lemma}
\begin{proof}
    Let $f : \pmone^n \to \pmone$ be arbitrary.
    Fix any $\eps \geq 0$.
    The smallest $M$ such that there exists a degree-$d$ $(M, \eps)$-SBQP-approximation to $f$ is given by the following linear program.
    \begin{align*}
        \min_{\substack{M \\ p \text{ of deg } < d}} \hspace*{1em} &M \\
        \text{s.t.} \hspace*{1em} & p(x) \geq f(x) - \eps & \forall x \in f^{-1}(1) \\
        & p(x) +  M \geq f(x) - \eps & \forall x \in f^{-1}(-1) \\
        & -p(x) \geq -f(x) - \eps & \forall x \in \pmone^n \\
        & M \geq 0
    \end{align*}
    The dual of this linear program is the following.
    \begin{align*}
        \max_{\psi^+, \psi^- : \pmone^n \to \R} \hspace*{1em} &\sum_{x \in \pmone^n} (\psi^+(x) - \psi^-(x))f(x) - \varepsilon \sum_{x \in \pmone^n} (\psi^+(x) + \psi^-(x))  \\
        \text{s.t.} \hspace*{1em} & \ang{\psi^+ - \psi^-, \chi_S} = 0 \hspace*{8em} \forall S \subseteq [n], |S| \leq d \\
        & \sum_{x : f(x) = -1} \psi^+(x) \leq 1 \\
        &\psi^-, \psi^+ \geq 0
    \end{align*}
    By a standard transformation (which appears, e.g., in the derivation of the dual formulation of approximate degree), any optimal solution to the dual program can be turned into one where $\psi^+$ and $\psi^-$ have disjoint support.
    Thus, we can set $\psi = \psi^+ - \psi^-$ to get the following equivalent dual linear program.
    \begin{align*}
        \max_{\psi : \pmone^n \to \R} \hspace*{1em} & \ang{\psi, f} - \varepsilon \norm{\psi}_1  \\
        \text{s.t.} \hspace*{1em} & \ang{\psi, \chi_S} = 0  &  \forall S \subseteq [n], |S| \leq d \\
        & \sum_{x : f(x) = -1, \psi(x) > 0} \psi(x) \leq 1
    \end{align*}
    By weak duality, to show that $f$ does not admit an $(M^*, \eps)$-SBQP-approximation it is sufficient to show that the dual program has objective at least $M^*$.
    Suppose we have a $\varphi$ as in the theorem statement.
    Consider $\psi = M^* \varphi / \eps$.
    This $\psi$ satisfies the first constraint of the dual since $\varphi$ has pure high degree at least $d$.
    It satisfies the second constraint since $\sum_{x : f(x) = -1, \psi(x) > 0} \psi(x) = \frac{M^*}{\eps} \sum_{x : f(x) = -1, \varphi(x) > 0} \varphi(x) \leq 1$.
    Finally, the objective attained by $\psi$ is $\frac{M^*}{\eps}\ang{\varphi, f} - \eps \cdot \frac{M^*}{\eps} \norm{\varphi}_1 \geq 2M^* - M^* = M^*$.
\end{proof}

Functions that satisfy the conditions in Lemma~\ref{l:sbqp-dual} are central to our proofs.
We call a $\varphi$ that has those properties with respect to a function $f$ and parameters $M$, $\eps$, and $d$ a \emph{relaxed $(M, \eps, d)$-SBQP dual to $f$}.
The word ``relaxed'' refers to the fact that while the existence of $\varphi$ is sufficient for a lower bound on the degree of an $(M, \eps)$-SBQP approximation, it might not be necessary.
This is because the proof of Lemma~\ref{l:sbqp-dual} loses tightness when it replaces $\psi$, which is obtained via LP duality, by the weaker (but easier to reason about) object $\varphi$.
Just as SBQP approximations are natural restrictions of one-sided approximations (allowing large but bounded error on one side), property~4 of the dual is a natural relaxation of the one-sidedness property of $\degminus$ duals described in Theorem~\ref{t:apdeg-dual}.
We use relaxed SBQP duals to prove that if a function $f$ has high  approximate degree, then all SBQP approximations to $\AND_m \circ f^m$ must have high degree.
To do so, we use the following well-known fact about the vanishing-error one-sided approximate degree of the $\AND$ function.

\begin{theorem}[\cite{buh-etal:small-error}] \label{t:and-apdeg}
    There exist universal constants $\delta_\AND$ and $\delta'_\AND$ such that for every $\gamma \in [2^{-m}, 1/3]$, $\delta_\AND \sqrt{m \log(1/\gamma)} < \degminus_\gamma(\AND_m) \leq \delta'_\AND \sqrt{m \log(1/\gamma)}$.
\end{theorem}

Strictly speaking, \cite{buh-etal:small-error} only implies $\apdeg_\gamma(\AND_m) = \Theta(\sqrt{m \log(1/\gamma)})$, but it is easy to see that one-sidedness is without loss of generality for the $\AND$ function. (This observation goes back to Gavinsky and Sherstov~\cite{gav-she:np-conp}.)
Suppose $\apdeg_\gamma(\AND_m) > d \geq 1$, and let $\psi$ be a dual witness to that lower bound guaranteed by Theorem~\ref{t:apdeg-dual}.
Since $\psi$ has pure high degree at least 1, we get that $0 = \ang{\psi, 1} = \psi(-\one_m) + \sum_{z \neq -\one_m} \psi(z)$.
That $\psi$ has correlation at least $\gamma$ with $\AND_m$ combined with the fact that $-\one_m$ is the only $-1$ instance of $\AND_m$ gives us $\gamma \leq \ang{\psi, \AND_m} = -\psi(-\one_m) + \sum_{z \neq -\one_m} \psi(z)$.
Combining the two equations yields $\psi(-\one_m) \leq -\gamma/2 < 0$, which means $\psi$ is a witness for $\degminus_\gamma(\AND_m) > d$.
The other ingredient in our construction is the following polynomial $p_k$.

\begin{lemma}[{\cite{she:direct-product}}] \label{l:sherstov-polynomial}
    Let $k$ be an even number in $[m-1]$, and define $p_k : [-1,1]^m \to \R$ as the unique degree-$k$ multilinear polynomial such that for all $z \in \pmone^m$,
    \begin{equation*}
        p_k(z) = \prod_{i = 1}^k (|z| - i).
    \end{equation*}
    Then $p_k(z) \geq 0$ for all $z \in [-1,1]^m$ and for every $\eta \in [0,1)$,
    \begin{equation*}
        \E_{z \sim \Pi(\eta)}[p_k(z)] \leq p_k(\one_m) (1 - \eta)^m \left\{ 1 + \frac{\eta^{k+1}}{(1-\eta)^m} {m \choose k + 1} \right\},
    \end{equation*}
    where $\Pi(\eta)$ is the distribution on $\pmone^m$ where each coordinate is $-1$ with probability $\eta$ and $1$ with probability $1 - \eta$.
\end{lemma}

We now prove our main connection between approximate degree and SBQP degree.
Our proof works by way of constructing a relaxed SBQP dual for $\AND_m \circ f^m$ from (1) a dual witness for the vanishing-error one-sided approximate degree of $\AND$, (2) a dual for the assumed approximate degree lower bound on $f$, and (3) the polynomials $p_k$.
The construction itself is a renormalized special case of the dual constructed in \cite[Theorem~6.1]{she:direct-product}.
Consequently, the correlation, $\ell_1$ norm, and pure high degree analyses for the dual follow from the corresponding analyses in that work.
Additionally, to get an SBQP dual, we need to show that the dual construction does not place too much mass on ``true'' inputs that it misclassifies.
We do so by taking advantage of the fact that the dual for $\AND$ is one-sided and keeping tighter track of how errors on misclassified ``true'' inputs accumulate.

\begin{theorem} \label{t:sbqp-lb}
    Let $n$ be a function of $m$, and let $\{f : \pmone^{n(m)} \to \pmone\}_{m \in \N}$ be a family of functions.
    Let $\gamma \in [2^{-m}, 1/3]$, $\eps > 0$, $M > 0$, and $k \in \N$ be such that $k$ is even and
    \begin{gather}
        \varepsilon \le \frac{\gamma}{2} - \frac{\eta^{k+1}}{(1-\eta)^m} \binom{m}{k+1} \text{ and} \label{eq:assm-eps-vs-gamma} \\
        \frac{\varepsilon}{M} \ge \frac{\eta^{k+1}}{2(1-\eta)^m} \binom{m}{k+1}. \label{eq:assm-eps-over-m}
    \end{gather}
    If $\apdeg_{1 - \eta}(f) > d \geq 1$ then
    for sufficiently large $m$, every $(M, \eps)$-SBQP approximation to $F = \AND_m \circ f^m$ has degree $\Omega(d \cdot  (\delta_\AND \sqrt{m\log(1/\gamma)}-k))$.
\end{theorem}

\begin{proof}
    Fix $n$, $\eta \in (0,1]$, $f : \pmone^n \to \pmone$, $d \geq 1$, and suppose $\apdeg_{1 - \eta}(f) > d$.
    Suppose $\gamma$, $\eps$, $M$, and $k$ are such that equations \ref{eq:assm-eps-vs-gamma} and \ref{eq:assm-eps-over-m} hold.
    By Lemma~\ref{l:sbqp-dual}, it suffices to construct a relaxed $(M, \eps, \Omega(d \cdot (\delta_\AND \sqrt{m \log(1/\gamma)} - k)) )$-SBQP dual $\Phi$ to $F$.

    By Theorems~\ref{t:apdeg-dual} and \ref{t:and-apdeg} there is a dual witness $\psi : \pmone^m \to \R$ for $\degminus_\gamma(\AND_m) \geq \delta_\AND \sqrt{m \log(1/\gamma)}$.
    Similarly, there is a dual $\varphi : \pmone^n \to \R$ witnessing $\apdeg_{1 - \eta}(f)$.
    Let $\widetilde{\mu}$ be the distribution on $\pmone^n$ with probability mass function $\widetilde{\mu}(y) = |\varphi(y)|$, and let $\mu$ be a distribution consisting of $m$ independent samples from $\widetilde{\mu}$.
    For every $z \in \pmone^m$ define $\mu_z$ to be $\mu$ conditioned on $(\ldots, \sgn \varphi(x_i), \ldots) = z$.

    Note that $\E_{y \sim \widetilde{\mu}}[\sgn \varphi(y)] = \sum_{y \in \pmone^n} \abs{\varphi(y)} \sgn \varphi(y) = \ang{\varphi, \chi_\emptyset}$. Since $\varphi$ has pure high degree $d \geq 1$, $\ang{\varphi, \chi_\emptyset} = 0$ and so $\Pr_{y \sim \widetilde{\mu}}[\varphi(y) > 0] = \Pr_{y \sim \widetilde{\mu}}[\varphi(y) < 0]$.
    In other words, if $y \sim \widetilde{\mu}$, then $\sgn \varphi(y)$ is a uniformly random bit.
    Consequently, if $x \sim \mu$, then $(\ldots, \sgn \varphi(x_i), \ldots)$ is uniformly distributed over $\pmone^m$.

    Since $\varphi$ is a dual for $\apdeg_{1-\eta}(f) > d$, $\ang{\varphi, f} \geq 1 - \eta$. Thus
    $1 - \eta \leq \sum_{y \in \pmone^n} f(y) \varphi(y) = \E_{y \sim \widetilde{\mu}}[f(y) \sgn \varphi(y)]$.
    Notice that for all $y$, $\frac{1 - f(y) \sgn \varphi(y)}{2}$ is equal to 1 if $\sgn \varphi(y) \neq f(y)$ and 0 otherwise. Thus
    \begin{equation*}
        \Pr_{y \sim \widetilde{\mu}}[\sgn \varphi(y) \neq f(y)]
        = \E_{y \sim \widetilde{\mu}} \left[ \frac{1 - f(y) \sgn \varphi(y)}{2} \right]
        \leq \frac{\eta}{2}.
    \end{equation*}
    Now define
    \begin{align*}
        \eta^{+} &= \Pr_{y \sim \widetilde{\mu}}[f(y) \neq \sgn \varphi(y) \mid \varphi(y) > 0], \\
        \eta^{-} &= \Pr_{y \sim \widetilde{\mu}}[f(y) \neq \sgn \varphi(y) \mid \varphi(y) < 0],
    \end{align*}
    and $\alpha : \pmone^n \to [-1, 1]$ as
    \begin{equation*}
        \alpha(y) = \begin{cases}
            \frac{1 - 2\eta + \eta^{-}}{1 - \eta^{-}} & \text{if } f(y) = \sgn \varphi(y) = -1 \\
            \frac{1 - 2\eta + \eta^{+}}{1 - \eta^{+}} & \text{if } f(y) = \sgn \varphi(y) = 1 \\
            -1 & \text{otherwise.}
        \end{cases}
    \end{equation*}
    The range of $\alpha$ is $[-1, 1]$ because $\max\{\eta^{+}, \eta^{-}\} \leq \eta$, which holds since $1 - \eta \leq 2 \Pr_{y \sim \widetilde{\mu}}[\sgn \varphi(y) \neq f(y)] = \frac{1}{2}(1-2\eta^{+}) + \frac{1}{2}(1-2\eta^{-})$.
    For $x \in (\pmone^n)^m$, we write $\alpha^m(x)$ to mean $(\ldots, \alpha(x_i), \ldots)$.
    Note that for $z \in \pmone^m$ and $i \in [m]$, $\E_{x \sim \mu_z}[\alpha(x_i)] = 1 - 2\eta$.

    We are now ready to present the dual.
    Let $p_k : [-1,1]^m \to [0,\infty)$ be the degree-$k$ multilinear polynomial given by Lemma~\ref{l:sherstov-polynomial}.
    Our relaxed SBQP dual $\Phi : (\pmone^n)^m \to \R$ is defined as
    \begin{equation*}
        \Phi(x_1, \ldots, x_m) = \frac{2^m}{p_k(\ldots, 1 - 2\eta, \ldots)} \cdot  \psi(\ldots, \sgn \varphi(x_i), \ldots) \cdot p_k(\alpha^m(x)) \prod_{i=1}^m \abs{\varphi(x_i)}.
    \end{equation*}

    As we now describe, the required correlation, $\ell_1$ norm, and pure high degree properties follow from the corresponding claims in \cite{she:direct-product}.
    First, \cite[Claim~6.3]{she:direct-product} tells us that
    \begin{align*}
        \ang{F, \Phi} &\geq \gamma - \frac{2 \eta^{k+1}}{(1-\eta)^m} {m \choose k + 1} \geq 2\eps,
    \end{align*}
    where the last inequality applies Eq.~\ref{eq:assm-eps-vs-gamma}.
    That $\norm{\Phi}_1 = 1$ follows from \cite[Claim~6.2]{she:direct-product}. Finally, \cite[Eq.~6.7]{she:direct-product} and the fact that $\psi$ has pure high degree at least $\delta_\AND \sqrt{m \log(1/\gamma)}$ implies that $\Phi$ has pure high degree $\Omega(d \cdot (\delta_\AND \sqrt{m \log(1/\gamma)} - k))$.

    It remains to show that $\Phi$ does not place too much mass on $-1$ instances that it ``misclassifies.''
    Notice that the only term in the definition of $\Phi$ that is not non-negative is $\psi(\ldots, \sgn \varphi(x_i), \ldots)$.
    For convenience, let $\beta = p_k(\ldots, 1-2\eta, \ldots)$.
    We can write the error on $-1$ instances as
    \begin{align*}
        &\sum_{x : F(x) = -1, \Phi(x) > 0} \Phi(x) = \sum_{x \in (\pmone^n)^m} \Phi(x) \cdot \mathbbm{1}[F(x) = -1 \land \Phi(x) > 0] \\
        &= \frac{2^m}{\beta} \sum_{x \in (\pmone^n)^m} \psi(\ldots, \sgn \varphi(x_i), \ldots) \cdot p_k(\alpha^m(x_i)) \cdot \mathbbm{1}[F(x) = -1 \land \Phi(x) > 0] \prod_{i=1}^m \abs{\varphi(x_i)} \\
        &= \frac{2^m}{\beta} \E_{x \sim \mu} \big[ \psi(\ldots, \sgn \varphi(x_i), \ldots) \cdot p_k(\alpha^m(x)) \cdot \mathbbm{1}[F(x) = -1 \land \Phi(x) > 0] \big] \\
        &= \frac{1}{\beta} \sum_{z \in \pmone^m} \E_{x \sim \mu_z} \big[ \psi(z) \cdot p_k(\alpha^m(x)) \cdot \mathbbm{1}[F(x) = -1 \land \psi(z) > 0] \big] \\
        &= \frac{1}{\beta} \sum_{z : \psi(z) > 0} \psi(z) \E_{x \sim \mu_z} \big[ p_k(\alpha^m(x)) \cdot \mathbbm{1}[F(x) = -1] \big].
        \numberthis \label{eq:onesided-error}
    \end{align*}
    The second-to-last equality holds since $\Pr_{x \sim \mu}[(\ldots, \sgn \varphi(x_i), \ldots) = z] = 2^{-m}$.
    Fix $z$.
    By multilinearity of $p_k$, $\E_{x \sim \mu_z}[p_k(\alpha^m(x))] = \beta$. Thus,
    \begin{align*}
        &\abs{\E_{x \sim \mu_z} \big[ p_k(\alpha^m(x)) \cdot \mathbbm{1}[F(x) = -1] \big] - \beta \cdot \mathbbm{1}[\AND_m(z) = -1]} \\
        &\hspace*{4em}= \abs{ \E_{x \sim \mu_z} \big[ p_k(\alpha^m(x)) \cdot \mathbbm{1}[F(x) = -1] - p_k(\alpha^m(x)) \cdot \mathbbm{1}[\AND_m(z) = -1] \big] } \\
        &\hspace*{4em}\leq \E_{x \sim \mu_z} \big[ \abs{p_k(\alpha^m(x))} \cdot \big\lvert \mathbbm{1}[F(x) = -1] - \mathbbm{1}[\AND_m(z) = -1] \big\rvert \big].
    \end{align*}
    Recall that $F(x) = \AND_m(\ldots, f(x_i), \ldots)$. So $\abs{\mathbbm{1}[F(x) = -1] - \mathbbm{1}[\AND_m(z) = -1]}$ is 1 when exactly one of $(\ldots, f(x_i), \ldots)$ and $z$ is equal to $-\one_m$ and 0 otherwise.
    By the non-negativity of $p_k$ on inputs in $[-1,1]^m$, it follows that
    \begin{align*}
        \Big\lvert \E_{x \sim \mu_z} \big[ p_k(\alpha^m(x)) \cdot &\mathbbm{1}[F(x) = -1] \big] - \beta \cdot \mathbbm{1}[\AND_m(z) = -1] \Big\rvert \\
        &\leq \E_{x \sim \mu_z} \big[ p_k(\alpha^m(x)) \cdot \mathbbm{1}[(\ldots, f(x_i), \ldots) \neq z] \big] \\
        &= \E_{x \sim \mu_z} \big[ p_k(\alpha^m(x)) \cdot \big( 1 - \mathbbm{1}[(\ldots, f(x_i), \ldots) = z] \big) \big] \\
        &= \beta - \E_{x \sim \mu_z} \big[ p_k(\alpha^m(x)) \cdot \mathbbm{1}[(\ldots, f(x_i), \ldots) = z] \big].
    \end{align*}
    Now \cite[Eq.~6.8]{she:direct-product} shows that \[\E_{x \sim \mu_z} \big[ p_k(\alpha^m(x)) \cdot \mathbbm{1}[(\ldots, f(x_i), \ldots) = z] \big] \geq (1 - \eta)^m p_k(\one_m),\] which means
    \begin{equation*}
        \abs{\E_{x \sim \mu_z} \big[ p_k(\alpha^m(x)) \cdot \mathbbm{1}[F(x) = -1] \big] - \beta \cdot \mathbbm{1}[\AND_m(z) = -1]} \leq \beta - (1 - \eta)^m p_k(\one_m).
    \end{equation*}
    Combining this bound with Eq.~\ref{eq:onesided-error}, we get
    \begin{align*}
        \sum_{x : F(x) = -1, \Phi(x) > 0} \Phi(x) &\leq \frac{1}{\beta} \sum_{z : \psi(z) > 0} \psi(z) \big( \beta \cdot \mathbbm{1}[\AND_m(z) = -1] + \beta - (1 - \eta)^m p_k(\one_m) \big) \\
        &= \left(\sum_{z : \psi(z) > 0} \psi(z) \right) \left(1 - \frac{(1 - \eta)^m p_k(\one_m)}{\beta} \right) \\
        &= \frac{1}{2} \left(1 - \frac{(1 - \eta)^m p_k(\one_m)}{\beta} \right).
    \end{align*}
    The second-to-last equality uses the one-sidedness of $\psi$: whenever $\AND_m(z) = -1$, $\psi(z) < 0$.
    The last equality uses the fact that $\psi$ is balanced.

    Now let $\Pi$ be a distribution on $\pmone^m$ where each coordinate is an independent, uniformly random element of $\pmone$.
    Then by multilinearity, $\E_{z \sim \Pi}[p_k(z)] = \beta$.
    By Lemma~\ref{l:sherstov-polynomial},
    \begin{align*}
        \frac{(1 - \eta)^m p_k(\one_m)}{\beta} &\geq \left( 1 + \frac{\eta^{k+1}}{(1-\eta)^m} {m \choose k + 1} \right)^{-1} \\
        &\geq 1 - \frac{\eta^{k+1}}{(1-\eta)^m} {m \choose k + 1}.
    \end{align*}
    Plugging this back in and using Eq.~\ref{eq:assm-eps-over-m},
    \begin{equation*}
        \sum_{x : F(x) = -1, \Phi(x) > 0} \Phi(x) \leq \frac{\eta^{k+1}}{2(1 - \eta)^m} {m \choose k + 1} \leq \frac{\eps}{M}. \qedhere
    \end{equation*}
\end{proof}

The ensuing corollaries apply Theorem~\ref{t:sbqp-lb} with specific parameters to obtain more ``user-friendly'' connections between the approximate degree of $f$ and the SBQP degree of $\AND_m \circ f^m$.

\begin{corollary} \label{c:large-error-to-sbqp}
    Let $n$ be a function of $m$, and let $\{f : \pmone^{n(m)} \to \pmone\}_{m \in \N}$ be a family of functions.
    Let $c > 0$ be a constant, $d \geq 1$, and $\eps \in [2^{-m-2}, 1/12]$ such that either $c \geq 1$ or $\eps = 2^{-\Omega(m)}$.
    If $\apdeg_{1 - \frac{1}{m^c}}(f) > d$, then there exists a constant $\tau > 0$ such that for sufficiently large $m$, every $(M, \eps)$-SBQP approximation to $F = \AND_m \circ f^m$ with $M \leq \eps 2^{\tau \sqrt{m \log(1/\eps)} \log m}$ must have degree $\Omega(d \sqrt{m \log(1/\eps)})$.
    The constant $\tau$ depends only on $c$ when $c \geq 1$, and on $c$ as well as the constant in the $\Omega$ when $c \in (0,1)$ but $\eps = 2^{-\Omega(m)}$.
\end{corollary}
\begin{proof}
    Let $n$, $f$, $c$, $d$, and $\eps$ be as in the corollary statement.
    Let $\tau > 0$ be a constant to be set later in this proof.
    We apply Theorem~\ref{t:sbqp-lb} with
    \begin{align*}
        \eta &= \frac{1}{m^c}, \\
        \gamma &= 4\eps, \\
        M &= \eps 2^{\tau \sqrt{m \log(1/\eps)} \log m}, \text{ and} \\
        k &= \text{the least even number greater than } \frac{\delta_\AND \sqrt{m \log(1/\gamma)}}{2}.
    \end{align*}
    This will imply that every $(M, \eps)$-SBQP approximation to $\AND_m \circ f^m$ has degree at least $\Omega(d \cdot (\delta_\AND \sqrt{m \log(1/\gamma)} - k)) = \Omega(d \sqrt{m \log(1/\eps)})$.
    The lower bound extends to smaller $M$ immediately since an $(M', \eps)$-SBQP approximation with $M' < M$ is also an $(M, \eps)$-SBQP approximation.

    We now get an upper bound on the quantity $\frac{\eta^{k+1}}{(1-\eta)^m} {m \choose k+1}$, which will allow us to show that the premises of Theorem~\ref{t:sbqp-lb} are satisfied.
    First, suppose $c \geq 1$.
    We have $(1 - \eta)^m = ((1 - 1/m^c)^{m^c})^{m^{1-c}}$.
    For large $m$, this is at least
    $(1/2e)^{m^{1-c}} \geq 1/2e$.
    By Theorem~\ref{t:and-apdeg} there is a universal constant $C$ such that $k \geq C \sqrt{m \log(1/\gamma)}$.
    We have
    \begin{align*}
        \frac{\eta^{k+1}}{(1-\eta)^m} {m \choose k + 1} &\leq 2e \cdot m^{-c\,(k+1)} \cdot m^{k+1} = 2e m^{-(c-1)(k+1)} = 2^{-\Omega(k \log m)}.
    \end{align*}
    The last bound holds because $c \geq 1$ and hence $m^{-(c-1)(k+1)} = 2^{-\Omega(k \log m)}$.

    Now suppose $c \in (0,1)$ but $\eps = 2^{-\Omega(m)}$.
    Since $\gamma = 4\eps$, there is a constant $C_1 \in (0,1]$ such that $\gamma = 2^{-C_1m}$.
    Also, $k + 1 \leq C_2 m$ for an appropriate constant $C_2 \in (0,1)$ that depends on $\delta_\AND$ and $C_1$.
    Using Stirling's approximation,
    \begin{align*}
        \frac{\eta^{k+1}}{(1-\eta)^m} {m \choose k + 1} &\leq \frac{m^{-c\,(k+1)}}{(1-m^{-c})^m} \left( \frac{me}{k+1} \right)^{k+1} \\
        &= \frac{(m^c)^{m-(k+1)}}{(m^c - 1)^m} \left( \frac{e}{C_2} \right)^{C_2 m} \\
        &= \exp_2 \left( (1-C_2) c m \log m - m \log(m^c - 1) + C_2 m \log(e/C_2) \right) \\
        &= 2^{-\Omega(m \log m)} = 2^{-\Omega(k \log m)},
    \end{align*}
    where the last bound holds since for large $m$, $(1 - C_2) cm \log m < m \log(m^c + 1)$.
    
    Thus in either case, it follows from $k \geq \frac{\delta_\AND}{2} \sqrt{m \log(1/\gamma)}$ that for sufficiently large $m$, $\frac{\eta^{k+1}}{(1-\eta)^m} {m \choose k + 1} \leq \gamma / 4$.
    Thus $\frac{\gamma}{2} - \frac{\eta^{k+1}}{(1-\eta)^m} {m \choose k + 1} \geq \frac{\gamma}{4} = \eps$, which establishes Eq.~\ref{eq:assm-eps-vs-gamma}.
    
    Eq.~\ref{eq:assm-eps-over-m} holds similarly. Since $k \geq \frac{\delta_\AND}{2} \sqrt{m \log(1/\gamma)} \geq C_3 \sqrt{m \log(1/\eps)}$ (for an appropriate universal constant $C_3$), there exists $\tau'$ such that
    \begin{align*}
        \frac{\eta^{k+1}}{2(1-\eta)^m} {m \choose k + 1} &\leq 2^{-\tau' C_3 \sqrt{m \log(1/\eps)} \log m}.
    \end{align*}
    Setting $\tau = \tau' \cdot C_3$, this quantity is at most $\eps/M$.
\end{proof}

\begin{corollary} \label{c:const-error-to-sbqp}
    Let $n$ be a function of $m$, and let $\{f : \pmone^{n(m)} \to \pmone\}_{m \in \N}$ be a family of functions.
    If $\apdeg_{1/3}(f) > d \geq 1$, then for sufficiently large $m$, every $(2^{m-1}, 2^{-m-2})$-SBQP approximation to $F = \AND_m \circ f^m$ must have degree $\Omega(d m)$.
\end{corollary}
\begin{proof}
    By standard error reduction techniques (see, e.g.,~\cite[Section~3.4]{bun-tha:approximate-degree}), it is possible to transform any degree-$k$ $0.99$-approximation to $f$ into a degree-$O(k)$ $1/3$-approximation to $f$.
    Thus it suffices to prove the corollary assuming $\apdeg_{0.99}(f) > d \geq 1$; the corollary as stated follows with at the cost of a constant factor in the degree lower bound.
    
    We use the well-known fact that $\degminus_{2^{-m}}(\AND_m) \geq m$ (note the absence of a constant factor); the dual witness justifying this is simply the normalized parity function on $m$ bits.
    Let $f : \pmone^n \to \pmone$ be such that $\apdeg_{0.99}(f) > d \geq 1$.
    We apply Theorem~\ref{t:sbqp-lb} with
    \begin{align*}
        \eta &= \frac{1}{100}, \\
        \gamma &= 2^{-m}, \\
        \eps &= 2^{-m-2}, \\
        M &= 2^{m-1}, \text{ and} \\
        k &= \frac{m}{2} - 1.\footnotemark
    \end{align*}
    \footnotetext{We assume for simplicity of exposition that $k$ is an even natural number. The case where $m/2 - 1$ is not an even number follows by setting $k$ to the least even number greater than $m/2 - 1$.}
    By Stirling's approximation,
    \begin{align*}
        \frac{\eta^{k+1}}{(1 - \eta)^m} {m \choose k+1} &\leq \frac{1}{100^{k+1}} \cdot \frac{100^m}{99^m} \cdot \left( \frac{me}{k+1} \right)^{k+1} \\
        &= \frac{100^{m/2}}{99^m} \cdot (2e)^{m/2} \tag{$k + 1 = m/2$} \\
        &= \left( \frac{\sqrt{200e}}{99} \right)^m \leq 2^{-2m}. \tag{$\sqrt{200e}/99 \leq 1/4$}
    \end{align*}
    Thus for large $m$,
    \begin{gather*}
        \frac{\gamma}{2} - \frac{\eta^{k+1}}{(1 - \eta)^m} {m \choose k+1} \geq 2^{-m-1} - 2^{-2m} \geq 2^{-m-2} = \eps \text{ and} \\
        \frac{\eta^{k+1}}{2(1 - \eta)^m} {m \choose k+1} \leq 2^{-2m-1} = \frac{\eps}{M},
    \end{gather*}
    satisfying the premises of Theorem~\ref{t:sbqp-lb}.
\end{proof}

%% file: sections/qma-dt-batch-ver.tex
In this section, we apply Corollaries~\ref{c:large-error-to-sbqp} and \ref{c:const-error-to-sbqp} to obtain various \QMAdt{} lower bounds.
Our proofs rely on the following connection between \QMAdt{} protocols and SBQP approximations.

\begin{lemma} \label{l:qma-to-sbqp}
    If $f : \pmone^n \to \pmone$ has a \QMAdt{} protocol with witness length $w$ and query complexity $q$, then for every $\eps > 0$ there is an $O(q \cdot \max\{w, \log(1/\eps)\})$-degree polynomial that $(2^{w+2}, \eps)$-SBQP approximates $f$.
\end{lemma}

\begin{proof}
    Suppose $f$ has a \QMAdt{} protocol with witness length $w$ and query complexity $q$.
    Fix $\eps > 0$, and let $\nu = \min\{\frac{\eps}{2}, 2^{-2w}\}$.
    By Theorem~\ref{t:marriott-watrous}, there is a \QMAdt{} protocol $P$ for $f$ with completeness and soundness error at most $\nu$, witness length $w$, and query complexity $O(q\log(1/\nu))$.
    Let $Q$ be the witness-free protocol that simulates $P$ with the maximally mixed state as the witness and accepts if and only if $P$ outputs $-1$.
    By soundness of $P$, whenever $f(x) = 1$, $\Pr_Q[Q^x$ accepts$] \leq \nu$.
    On the other hand, on inputs where $f(x) = -1$, the probability that $Q$ accepts is at least $(1 - \nu) 2^{-w} \geq (1 - 2^{-2w}) 2^{-w} \geq 2^{-w-1}$.

    Let $\widetilde{p}$ be the $O(q \log(1/\nu))$-degree polynomial computing the acceptance probability of $Q$. Define $p = 1 - 2 \cdot 2^{w+1} \widetilde{p}$.
    The degree of $p$ is $O(q \log(1/\nu)) = O(q \cdot \max\{w, \log(1/\eps)\})$.
    \begin{enumerate}
        \item If $f(x) = -1$ then $2^{w+1} \widetilde{p}(x) \in [1, 2^{w+1}]$ and so $p(x) \in [1 - 2^{w+2}, -1] \subset [-2^{w+2} - 1 - \eps, -1 + \eps]$.
        \item If $f(x) = 1$ then $2^{w+1} \widetilde{p}(x) \in [0, \nu]$ and so $p(x) \in [1 - 2\nu, 1] \subset [1 - \eps, 1 + \eps]$.
    \end{enumerate}
\end{proof}

\subsection{The \QMAdt{} Complexity of Batch Verification}

Our first lower bound exhibits a smooth tradeoff between witness length and query complexity of batch verifying functions with high large-error approximate degree.

\begin{theorem} \label{t:qma-dt-tradeoffs-c-big}
    Let $n$ be a function of $m$, and $\{f : \pmone^{n(m)} \to \pmone\}_{m \in \N}$ be an arbitrary family of functions.
    For all sufficiently large $m$, if $\apdeg_{1 - 1/m}(f) > d \geq 1$ then every \QMAdt{} protocol for $\AND_m \circ f^m$ with witness length $w = O(m)$ requires $\Omega(d \sqrt{m/w})$ queries.
\end{theorem}
\begin{proof}
    Fix $n(m)$ and $f : \pmone^{n(m)} \to \pmone$.
    Assume $\apdeg_{1-1/m}(f) > d \geq 1$.
    Suppose there is a \QMAdt{} protocol for $\AND_m \circ f^m$ with witness length $w = O(m)$ and query complexity $q$.
    We assume that $w \geq \log(12)$; the case where $w < \log(12)$ follows immediately with a smaller constant in the $\Omega$.
    Let $\eps = \max\{2^{-w}, 2^{-m-2}\} \in [2^{-m-2}, 1/12]$.
    The premises of Corollary~\ref{c:large-error-to-sbqp} are satisfied with $c = 1$.
    Thus, by Corollary~\ref{c:large-error-to-sbqp}, there exists a constant $\tau$ such that every $(M, \eps)$-SBQP approximation to $\AND_m \circ f^m$ with $M \leq \eps 2^{\tau \sqrt{m \log(1/\eps)} \log m}$ has degree $\Omega(d \sqrt{m \log(1/\eps)})$.

    By Lemma~\ref{l:qma-to-sbqp}, there is a degree-$O(qw)$ polynomial $p$ that $(2^{w+2}, \eps)$-SBQP approximates $\AND_m \circ f^m$.
    Since $\log(1/\eps) = \Omega(w)$, we have $\tau \sqrt{m \log(1/\eps)} \log m = \Omega(\sqrt{mw} \log m)$.
    Thus for large $m$, $2^{w+2} \leq \eps 2^{\tau \sqrt{m\log(1/\eps)} \log m}$.
    So $O(qw) = \deg(p) = \Omega(d \cdot \sqrt{m \log(1/\eps)}) = \Omega(d \sqrt{mw})$.
    Rearranging yields the desired lower bound.
\end{proof}

Observe that the proof of Theorem~\ref{t:qma-dt-tradeoffs-c-big} can actually tolerate witness lengths all the way up to $\tau' m \log m$ for some fixed constant $\tau' \in (0,1)$.
Roughly speaking, it shows that whenever $w \leq \tau' m \log m$, any \QMAdt{} protocol for $\AND_m \circ f^m$ with witness length $w$ must have query complexity at least $\Omega(\apdeg_{1 - 1/m}(f) \sqrt{m/w})$.
Combining this observation with Theorem~\ref{t:cnf-high-odeg} we show that even a constant-factor improvement in witness length over the trivial protocol for batch verifying DNFs incurs a large increase in query complexity.
\begin{corollary}
    For all sufficiently large $n$ and all $\delta > 0$ there exists an explicitly given function $f : \pmone^n \to \pmone$ computable by a polynomial-size, constant-width DNF and a constant $k \in \N$ such that for all $m = n^{O(1)} \leq n$, the following hold.
    \begin{enumerate}
        \item There is a \QMAdt{} protocol for $\AND_m \circ f^m$ with witness length $k m \log n$ and query complexity $O(k\sqrt{m})$.
        \item There is a constant $\tau' \in (0,1)$ such that every \QMAdt{} protocol for $\AND_m \circ f^m$ with witness length at most $\tau' m \log n$ must have query complexity $\min\{\Omega(n^{1-\delta} \sqrt{m/w}), \Omega(n^{1-\delta} m/w)\}$.
    \end{enumerate}
\end{corollary}
\begin{proof}
    By Theorem~\ref{t:cnf-high-odeg} for every $n$ and $\delta > 0$ there exists an explicitly given function $f : \pmone^n \to \pmone$ computable by a polynomial-size, constant-width DNF such that $\apdeg_{1 - 1/n}(f) > n^{1-\delta}$.
    Let $k \in \N$ be a constant such that the DNF computing $f$ has width at most $k$ and size a most $n^k$.
    Henceforth we refer to $f$ and the DNF computing it interchangeably.
    Let $m = n^a$ for some constant $a \in (0,1]$.
    Consider the following \QMAdt{} protocol for $\AND_m \circ f^m$.
    \begin{enumerate}
        \item The witness consists of indices $i_1, i_2, \ldots, i_m \in ([n^k])^m$ such that the $i_j$th term in the $j$th copy of $f$ is satisfied.
        \item The verifier, given query access to $x$ and the witness $(i_1, i_2, \ldots, i_m)$ uses Grover's search algorithm to look for $j \in [m]$ such that the $i_j$th term of the $j$th copy of $f$ is not satisfied.
        \item If such an index is found, the verifier rejects. Else, it accepts.
    \end{enumerate}
    If $(\AND_m \circ f^m)(x) = -1$ then the verifier, given the honest witness, will not find an unsatisfied term among those indicated by the witness. Thus the protocol has perfect completeness.
    If $(\AND_m \circ f^m)(x) = 1$ then there must be some $j$ such that the $j$th copy of $f$ is not satisfied. It follows that none of the terms of $f$ are satisfied in that copy, and so given any witness $(i_1, i_2, \ldots, i_m)$ the $i_j$th term of the $j$th copy of $f$ will not be satisfied.
    By properties of Grover search, the verifier finds that index and rejects with probability $2/3$. Thus the protocol is sound.
    The witness length of the protocol is at most $k m \log n$ since each term requires $k \log n$ bits to describe.
    Whether or not each term is satisfied can be computed exactly with $k$ queries.
    One can thus implement the search over these checks with $O(k \sqrt{m})$ queries using standard techniques (see, e.g., \cite{hoy-mos-dew:quantum-search}).
    This proves the upper bound.

    The lower bound holds due to the proof of Theorem~\ref{t:qma-dt-tradeoffs-c-big}.
    Since $m \leq n$, $\apdeg_{1-1/m}(f) \geq \apdeg_{1-1/n}(f) > n^{1-\delta}$.
    For convenience, let $d = n^{1-\delta}$.
    Setting $\eps = \max\{2^{-w}, 2^{-m-2}\}$, we can apply Corollary~\ref{c:large-error-to-sbqp} as in the proof of Theorem~\ref{t:qma-dt-tradeoffs-c-big}.
    This yields a constant $\tau$ such that every $(M, \eps)$-SBQP approximation to $\AND_m \circ f^m$ with $M \leq \eps 2^{\tau \sqrt{m \log(1/\eps)} \log m}$ must have degree $\Omega(d \sqrt{m \log(1/\eps)}) = \Omega(d \sqrt{m \cdot \min\{m,w\}})$.
    Due to our setting of $\eps$ and $m$, there exists a constant $\tau'' > 0$ such that the degree lower bound holds for all $(M, \eps)$-SBQP approximations with $M \leq 2^{\tau'' a \sqrt{m \cdot \min\{m - w\}} \log n}$.
    Set $\tau' = \tau'' a/2$.
    Suppose there is a \QMAdt{} protocol for $\AND_m \circ f^m$ with witness length $w \leq \tau' m \log n$.
    Then by Lemma~\ref{l:qma-to-sbqp} there is a $(2^{w+2}, \eps)$-SBQP approximation to $\AND_m \circ f^m$ with degree $O(q \min\{w, \max\{w, m\}\}) = O(qw)$.
    Combining the upper and lower bounds, $q = \Omega(d \sqrt{m \cdot \min\{m,w\}} / w) = \min\{\Omega(d\sqrt{m/w}), \Omega(dm/w)\}$.    
\end{proof}

Our techniques also allow us to get query complexity lower bounds for batch verifying $f$ from lower bounds on $\apdeg_{1 - 1/m^c}(f)$ for arbitrarily small constants $c > 0$ and on $\apdeg_{1/3}(f)$.

\begin{theorem} \label{t:qma-dt-lb-small-error}
    Let $n$ be a function of $m$ and let $\{f : \pmone^{n(m)} \to \pmone\}_{m \in \N}$ be an arbitrary family of functions. The following statements hold for all sufficiently large $m$.
    \begin{enumerate}
        \item For every $c > 0$, if $\apdeg_{1-1/m^c}(f) > d \geq 1$ then every \QMAdt{} protocol for $\AND_m \circ f^m$ with witness length at most $O(m)$ requires $\Omega(d)$ queries.
        \item If $\apdeg_{1/3}(f) > d \geq 1$ then every \QMAdt{} protocol for $\AND_m \circ f^m$ with witness length at most $m - 3$ requires $\Omega(d)$ queries.
    \end{enumerate}
\end{theorem}
\begin{proof} \textbf{(Part 1)} Fix $c > 0$ and $f : \pmone^{n(m)} \to \pmone$. Assume $\apdeg_{1 - 1/m^c}(f) > d \geq 1$, and suppose there is a \QMAdt{} protocol for $\AND_m \circ f^m$ with witness length $w = O(m)$ and query complexity $q$. Assume that $w \geq \log(12)$; the complementary case follows immediately.
    Let $\eps = 2^{-m-2}$.
    By Corollary~\ref{c:large-error-to-sbqp} there is a constant $\tau$ that depends only on $c$ such that every $(M, \eps)$-SBQP approximation to $\AND_m \circ f^m$ with $M \leq \eps 2^{\tau \sqrt{m \log(1/\eps)} \log m}$ must have degree $\Omega(d\sqrt{m \log(1/\eps)}) = \Omega(dm)$.

    By Lemma~\ref{l:qma-to-sbqp} there is a degree-$O(qm)$ polynomial that $(2^{w+2}, \eps)$-SBQP approximates $\AND_m \circ f^m$.
    We have $\eps 2^{\tau \sqrt{m \log(1/\eps)} \log m} = 2^{\Omega(m \log m)} \geq 2^{O(m)} = 2^{w+2}$.
    So we must have $O(qm) = \deg(p) = \Omega(dm)$, which implies $q = \Omega(d)$ as claimed.
    
    \textbf{(Part 2)} Suppose $\apdeg_{1/3}(f) > d \geq 1$, and suppose there is a \QMAdt{} protocol for $\AND_m \circ f^m$ with witness length $w \leq m-3$ and query complexity $q$.
    By Corollary~\ref{c:const-error-to-sbqp}, for sufficiently large $m$, every $(M, 2^{-m-2})$-SBQP approximation with $M \leq 2^{m-1}$ must have degree $\Omega(dm)$.
    By Lemma~\ref{l:qma-to-sbqp} there is a degree-$O(qm)$ polynomial that $(2^{w+2}, 2^{-m-2})$-SBQP approximation to $\AND_m \circ f^m$.
    Since $w \leq m - 3$, $2^{w+2} \leq 2^{m-1}$. The query lower bound follows.
\end{proof}

\subsection{\QMAdt{} Lower Bounds for Natural Functions}

Some functions have the property that they contain as sub-functions the problem of batch verifying smaller copies of themselves.
For such functions, a lower bound on the \QMAdt{} complexity of batch verification immediately implies a lower bound on the \QMAdt{} complexity of the functions themselves.
In this subsection, we use this insight along with Theorem~\ref{t:qma-dt-lb-small-error} to get \QMAdt{} lower bounds for element distinctness, surjectivity, and $k$-distinctness.
Our lower bounds make use of the observation by Sherstov and Thaler~\cite{she-tha:vanishing-error} that the surjectivity and $k$-element distinctness functions contain the batch verification problem for smaller instances of themselves as a subfunction.

\begin{theorem} \label{t:qma-dt-lower-bounds} ~
    \begin{enumerate}
        \item \label{item:and-or-lb} $\QMA^\dt(\AND_m \circ \OR_n^m) = \Omega((mn)^{1/3})$. Consequently, $\QMA^\dt(\AND_{n^{1/3}} \circ \OR_{n^{2/3}}^{n^{1/3}}) = \Omega(n^{1/3})$.
        \item \label{item:surj-lb} $\QMA^\dt(\SURJ_n) = \widetilde{\Omega}(n^{3/7})$.
        \item \label{item:ked-lb} $\QMA^\dt(k\ED_n) = \widetilde{\Omega}(n^{\frac{3k-1}{7k-1}})$.
    \end{enumerate}
\end{theorem}

Before presenting the proof of our lower bounds, we briefly mention lower bounds that follow directly from prior work.
For $\AND_{n^{1/3}} \circ \OR_{n^{2/3}}^{n^{1/3}}$, the best known lower bound is $\Omega(n^{1/6})$, which follows---via the connection between \QMAdt{} complexity and threshold degree in~\cite{dal-etal:quantum-proofs}---from the fact that the threshold degree of $\AND_{n^{1/3}} \circ \OR_{n^{2/3}}^{n^{1/3}}$ is $\Omega(n^{1/3})$.
For $\SURJ_n$, the best known lower bound is $\widetilde{\Omega}(n^{1/4})$, due to a threshold degree lower bound of $\widetilde{\Omega}(\sqrt{n})$~\cite{bun-tha:large-ac-0}.
Finally, for $k\ED_n$, we have $\QMA^\dt(k\ED_n) \geq \QMA^\dt(\ED_n) = \widetilde{\Omega}(n^{2/5})$, which was proved implicitly in~\cite{she-tha:vanishing-error}.

\begin{proof}[Proof of Theorem~\ref{t:qma-dt-lower-bounds}]~

    \textbf{(Part \ref{item:and-or-lb})}
    Note that for every $m'$ that divides $m$, $\AND_m \circ \OR_n^m = \AND_{m'} \circ (\AND_{m/m'} \circ \OR_n^{m/m'})^{m'}$.
    By a lower bound from Section~\ref{sec:prelims}, $\apdeg(\AND_{m/m'} \circ \OR_n^{m/m'}) = \Omega(\sqrt{mn/m'})$.
    By Theorem~\ref{t:qma-dt-lb-small-error}, every \QMAdt{} protocol for $\AND_{m'} \circ (\AND_{m/m'} \circ \OR_n^{m/m'})^{m'}$ with witness length at most $m' - 3$ must have query complexity $\Omega(\sqrt{mn/m'})$. Setting $m' = (mn)^{1/3}$, we get that every \QMAdt{} protocol or $\AND_m \circ \OR_n^m$ either has witness length at least $(mn)^{1/3} - 3$ or query complexity $\Omega(\sqrt{mn/(mn)^{1/3}}) = \Omega((mn)^{1/3})$, which means $\QMA^\dt(\AND_m \circ \OR_n^m) = \Omega((mn)^{1/3})$.

    \textbf{(Part \ref{item:surj-lb})}
    Set $R = N/2$ as is typical.
    Observe that for all $m$ that divide $N/2$, $\QMA^\dt(\SURJ_{N,R}) \geq \QMA^\dt(\AND_m \circ \SURJ_{N/m, R/m}^m)$.
    This is due to the fact that evaluating $\AND_m \circ \SURJ_{N/m, R/m}^m$ amounts to evaluating whether $m$ different lists of numbers from $[R/m]$ each contain every $r \in [R/m]$.
    One way to do this is to, for all $i \in [m]$, add $(i-1)R/m$ to each number in the $i$th list, concatenate all the lists, and then evaluate $\SURJ_{N,R}$ on the concatenated list.
    It follows that any \QMAdt{} protocol for $\SURJ_{N,R}$ can be used to solve $\AND_m \circ \SURJ_{N/m, R/m}^m$ with no additional cost.
    Recall that $\apdeg(\SURJ_{N/m, R/m}) = \widetilde{\Omega}((N/m)^{3/4})$.
    By Theorem~\ref{t:qma-dt-lb-small-error}, every \QMAdt{} protocol for $\AND_m \circ \SURJ_{N/m, R/m}^m$ with witness length at most $m-3$ must have query complexity $\widetilde{\Omega}((N/m)^{3/4})$.
    Setting $m = N^{3/7}$, every \QMAdt{} protocol for $\AND_m \circ \SURJ_{N/m, R/m}^m$ must either have witness length at least $N^{3/7} - 3$ or query complexity at $\widetilde{\Omega}(N^{3/7})$.
    Hence $\QMA^\dt(\SURJ_n) = \widetilde{\Omega}(N^{3/7})$, which is $\widetilde{\Omega}(n^{3/7})$.

    \textbf{(Part \ref{item:ked-lb})}
    Like in the case of surjectivity, $\QMA^\dt(k\ED_{N,N}) \geq \QMA^\dt(\AND_m \circ k\ED_{N/m, N/m}^m)$. When one shifts the input to $\AND_m \circ k\ED_{N/m, N/m}^m$ as in part~\ref{item:surj-lb}, the shifted input is a ``yes'' instance of $k\ED_{N,N}$ if and only if the original instance is a ``yes'' instance of $\AND_m \circ k\ED_{N/m, N/m}^m$.
    Recall that $\apdeg(k\ED_{N/m, N/m}) = \widetilde{\Omega}((N/m)^{3/4 - 1/(4k)})$.
    So Theorem~\ref{t:qma-dt-lb-small-error} tells us that every \QMAdt{} protocol for $\AND_m \circ k\ED_{N/m,N/m}^m$ with witness length at most $m - 3$ must have query complexity $\widetilde{\Omega}((N/m)^{3/4 - 1/(4k)})$. Setting $m = n^{\frac{3k-1}{7k-1}}$ yields the desired lower bound.
\end{proof}

%% file: sections/qma-cc-batch-ver.tex
In this section we show that our lower bound technique lifts to the communication model.
To do so, we prove a modified version of the pattern matrix method~\cite{she:pattern-matrix} for QMA.
Our analysis builds on prior generalizations of the pattern matrix method by Gavinsky and Sherstov~\cite{gav-she:np-conp}, Klauck~\cite{kla:arthur-merlin-communication}, and Sherstov~\cite{she:dnf-cnf}.

\begin{lemma}[{\cite[Lemma 10]{lin-shr:factorization}}] \label{lem:linial-shraibman}
    Let $P$ be a quantum communication protocol (with or without prior entanglement) with cost $c$ on input sets $X, Y$ which outputs values in $\zo$. Then there exist real matrices $A$ and $B$ with $\|A\|_F \le 2^c \sqrt{|X|}$ and $\|B\|_F \le 2^{c} \sqrt{|Y|}$ such that
    \begin{equation*}
        \big[ \E_P[P(x, y)] \big]_{x \in X, y \in Y} = AB.
    \end{equation*}
\end{lemma}

\begin{lemma} \label{l:pattern-matrix-method}
    Let $F : X \times Y \to \pmone$ and let $\Psi = [\Psi_{x,y}]_{x \in X, y \in Y}$ be a real matrix such that $\norm{\Psi}_1 = 1$.
    Let
    \begin{align*}
        \alpha &= \sum_{x \in X, y \in Y : F(x, y) = -1 \land \Psi(x, y) \leq 0} |\Psi(x, y)|, \\
         \beta &= \sum_{x \in X, y \in Y : F(x, y) = -1 \land \Psi(x, y) > 0} |\Psi(x, y)|, \\
         Q &= \log \frac{\alpha}{\|\Psi\| \sqrt{|X||Y|}}.
    \end{align*}
    Then for every \QMAcc{} protocol with witness length $w$ and communication cost $c$ computing $F$, if $w \leq \log(\alpha/\beta) - 2$ then $c \geq \min\{\Omega(Q/w), \Omega(Q/\log(1/\alpha))\}$.
\end{lemma}
\begin{proof}
    Let $P$ be any \QMAcc{} protocol with witness length $w$, communication $c$, and error $1/3$ computing $F$.
    Assume $w \leq \log(\alpha/\beta) - 2$.
    We will show that $c \geq \min\{\Omega(Q/w), \Omega(Q/\log(1/\alpha))\}$.
    Let $\delta = \alpha 2^{-w-2} \leq 1/4$.
    By (the communication analog of) Theorem~\ref{t:marriott-watrous}, there exists a \QMAcc{} protocol $P'$ with witness length $w$ and communication $c' = O(c \log(1/\delta))$ that computes $F$ with completeness and soundness error at most $\delta$.
    Let $P''$ be the ordinary (witness-free) quantum communication protocol with communication $c'$ that simulates $P'$ with the maximally mixed state as the witness, outputs 1 if $P'$ outputs $-1$, and outputs 0 if $P'$ outputs 1 (i.e., converts the output to the $\zo$ basis). Then
    \begin{align*}
         F(x, y) = -1 &\implies \E_{P''} [P''(x, y)] \in [2^{-w}(1-\delta), 1], \text{ and} \\
         F(x, y) = 1 &\implies \E_{P''} [P''(x, y)] \in [0, \delta].
    \end{align*}
    Let $\Pi = \big[ \E_{P''}[P''(x, y)] \big]_{x \in X, y \in Y}$. Then
    \begin{align*}
        \ang{\Pi, -\Psi} &\ge \alpha \cdot 2^{-w}(1-\delta) - \beta - (1-\alpha-\beta)\delta \\
        &\ge \frac{3}{4}\alpha \cdot 2^{-w} - \beta - \delta \\
        &= \alpha \cdot 2^{-w-1} - \beta \\
        &\ge \alpha \cdot 2^{-w-2},
    \end{align*}
    where the last inequality follows from our assumption that $w \le \log(\alpha/\beta) - 2$. Meanwhile, by Lemma~\ref{lem:linial-shraibman}, there are matrices $A, B$ with $AB = \Pi$ and $\|A\|_F \|B\|_F \le 4^{c'} \sqrt{|X||Y|}$. Thus
    \begin{align*}
        \ang{\Pi, -\Psi} &\le \|\Pi\|_{\Sigma} \cdot \|\Psi\| \\
        &\le \|A\|_F \|B\|_F \cdot \|\Psi\| \\
        &\le 4^{c'} \|\Psi\| \sqrt{|X||Y|}.
    \end{align*}
    Combining the two inequalities yields $2c' + w + 2 \ge Q$, and hence $c \ge \min\{\Omega(Q/w), \Omega(Q / \log(1/\alpha))\}$.
\end{proof}

The functions we obtain \QMAcc{} lower bounds for are obtained by block composing hard query problems with a constant-size ``indexing gadget.''
Functions composed with indexing gadgets are closely related to their ``pattern matrices.''
These matrices, combined with the relaxed SBQP duals constructed in the previous section, allow us to apply Lemma~\ref{l:pattern-matrix-method} to obtain a communication complexity analog of Theorems~\ref{t:qma-dt-tradeoffs-c-big} and \ref{t:qma-dt-lb-small-error}.

\begin{definition}
    The \emph{indexing gadget on $k$ bits} is the function $g_k : \pmone^k \times ([k] \times \pmone) \to \pmone$ defined as $g(x, (j,w)) = x_j \oplus w$.
\end{definition}

\begin{definition}
    The $k$-Pattern Matrix of a function $\psi : \pmone^n \to \R$ is the $2^{nk} \times (2k)^n$ matrix $\mathrm{PM}_k(\psi)$ given by
    \begin{align*}
        \mathrm{PM}_k(\psi)[(x^1, \ldots, x^n), ((j_1, w_1), \ldots, (j_n, w_n))] &= \psi(\ldots, g_k(x^i, (j_i, w_i)), \ldots) \\
        &= \psi(x|_J \oplus w)
    \end{align*}
    where $J = (j_1, \ldots, j_n)$ and $x|_J = (x^1_{j_1}, \ldots, x^n_{j_n})$.
    Note that when $\psi$ is a Boolean function, $\mathrm{PM}_k(\psi)$ is the communication matrix of $\psi \circ g_k^n$.
\end{definition}

\begin{lemma}[{\cite[Theorem~4.3]{she:pattern-matrix}}] \label{l:pm-spectral-norm}
    For every function $\psi : \pmone^n \to \R$,
    \begin{equation*}
        \norm{\mathrm{PM}_k(\psi)} = \sqrt{2^{nk} \cdot (2k)^n} \cdot \max_{S \subseteq [n]} \left( |\widehat{\psi}(S)| k^{-|S|/2} \right).
    \end{equation*}
\end{lemma}

\begin{theorem} \label{t:qma-cc-tradeoff-c-large}
    Let $k \geq 2$ be a constant, $n$ be a function of $m$, and $\{f : \pmone^n \to \pmone\}_{m \in \N}$ be a family of functions.
    If for sufficiently large $m$, $\apdeg_{1-1/m}(f) > d \geq 1$, then every \QMAcc{} protocol for $(\AND_m \circ f^m) \circ g^{nm}_k$ with witness length $w = O(m)$ must have communication cost $\Omega_k(d\sqrt{m/w})$.
\end{theorem}

\begin{proof}
    Fix $k$ and $f$ as in the theorem statement, and suppose there is a \QMAcc{} protocol for $(\AND_m \circ f^m) \circ g^{nm}_k$ with witness length $w = O(m)$.
    Assume $w \geq \log(12)$; the case where $w < \log(12)$ follows immediately by taking a smaller constant in the $\Omega$.
    For convenience, let $h = \AND_m \circ f^m$.
    Let $\eps = \max\{2^{-w}, 2^{-m-2}\} \in [2^{-m-2}, 1/12]$.
    By Corollary~\ref{c:large-error-to-sbqp} there exists a constant $\tau$ such that every $(M, \eps)$-SBQP approximation to $\AND_m \circ h^m$ with $M = \eps 2^{\tau \sqrt{m \log(1/\eps)} \log m}$ has degree $\Omega(d \sqrt{m \log(1/\eps)}) = \Omega(d \sqrt{mw})$.
    Furthermore, its proof (via Theorem~\ref{t:sbqp-lb}) gives a relaxed $(M, \eps, \Omega(d\sqrt{mw}))$-SBQP dual $\Phi : \pmone^{mn} \to \R$ for $h$ which has the following properties.
    \begin{enumerate}
        \item $\ang{\Phi, h} > 2 \eps$,
        \item $\norm{\Phi}_1 = 1$,
        \item $\Phi$ has pure high degree at least $d' = \Omega(d \sqrt{mw})$, and
        \item $\sum_{x : h(x) = -1 \land \Phi(x) > 0} \Phi(x) \leq \eps/M = 2^{-\tau \sqrt{m\log(1/\eps)} \log m}$.
    \end{enumerate}

    Let us define $X = (\pmone^k)^{nm}$, $Y = ([k] \times \pmone)^{nm}$, $F : X \times Y \to \pmone$ as $F = h \circ g_k^{nm}$, and $\Psi : X \times Y \to \R$ as $\Psi = \mathrm{PM}_k(2^{-nmk} \cdot k^{-nm} \cdot \Phi) = 2^{-nmk} \cdot k^{-nm} \cdot (\Phi \circ g_k^{nm})$. Note that for every $z \in \pmone^{nm}$, $w \in \pmone^{nm}$, and $J \in [k]^{nm}$, there are exactly $2^{nm(k-1)}$ vectors $x \in (\pmone^k)^{nm}$ such that $x|_J \oplus w = z$. It follows that for every $z \in \pmone^{nm}$, there are $2^{nmk} \cdot k^{nm}$ pairs $(a,b)$ such that $g_k(a,b) = z$.
    Thus
    \begin{align*}
        \norm{\Psi}_1 &= 2^{-nmk} \cdot k^{-nm} \sum_{x \in X, y \in Y} |(\Phi \circ g_k^{nm})(x,y)| \\
        &= 2^{-nmk} \cdot k^{-nm} \sum_{z \in \pmone^n} 2^{nmk} \cdot k^{nm} |\Phi(z)| \\
        &= \norm{\Phi}_1 = 1.
    \end{align*}

    We now define $\alpha$, $\beta$, and $Q$ as in Lemma~\ref{l:pattern-matrix-method} and get appropriate bounds on them. First,
    \begin{align*}
        \beta &= \sum_{x \in X, y \in Y : F(x,y) = -1 \land \Psi(x,y) > 0} |\Psi(x,y)| \\
        &= \sum_{z : h(z) = -1 \land \Phi(z) > 0} |\Phi(z)| \\
        &\leq 2^{-\tau \sqrt{m \log(1/\eps)} \log m}.
    \end{align*}
    Next,
    \begin{align*}
        \alpha &= \sum_{x \in X, y \in Y : F(x,y) = -1 \land \Psi(x,y) \leq 0} |\Psi(x,y)| \\
        &= \sum_{z : h(z) = -1 \land \Phi(z) \leq 0} |\Phi(z)|.
    \end{align*}
    Let $A^- = \sum_{z : h(z) = 1 \land \Phi(z) \leq 0} |\Phi(z)|$ and $A^+ = \sum_{z : h(z) = 1 \land \Phi(z) > 0} |\Phi(z)|$.
    We have the following system of equations due to the definitions of correlation and L1 norm, and the fact that $\Phi$ is balanced.
    \begin{align*}
        \norm{\Phi}_1 &= \alpha + \beta + A^- + A^+, \\
        \ang{\Phi, h} &= \alpha - \beta - A^- + A^+, \text{ and} \\
        0 &= -\alpha + \beta - A^- + A^+.
    \end{align*}
    Simplifying this system yields
    \begin{equation*}
        \alpha = \frac{\ang{\Phi, h}}{2} + 2\beta > \frac{\ang{\Phi, h}}{2} > \eps.
    \end{equation*}
    Now we bound $\norm{\Psi}$. By definition of the Fourier transform, for all $S \subseteq [nm]$,
    \begin{equation*}
        |\widehat{\Phi}(S)| = 2^{-nm} \left| \sum_{x \in \pmone^{nm}} \Phi(x) \chi_S(x) \right| \leq 2^{-nm} \norm{\Phi}_1 = 2^{-nm}.
    \end{equation*}
    Applying Lemma~\ref{l:pm-spectral-norm} and using the fact that $\widehat{\Phi}(S) = 0$ for all $|S| \leq d'$ gives
    \begin{equation*}
        \norm{\Psi} < \left( 2^{nmk} (2k)^{nm} \right)^{1/2} 2^{-nmk} k^{-nm} 2^{-nm} k^{-d'/2}
        = \left( 2^{nmk} (2k)^{nm} \right)^{-1/2} k^{-d'/2}.
    \end{equation*}
    Hence,
    \begin{align*}
        Q &= \log \frac{\alpha}{\norm{\Psi} \sqrt{|X||Y|}} \\
        &\geq \log \eps - \frac{1}{2}\log(|X||Y|) + \frac{nmk}{2} + \frac{nm \log(2k)}{2} + \frac{d'}{2} \log k \\
        &= \log \eps + \frac{d'}{\log k} \\
        &= - w + \frac{d \sqrt{mw}}{\log k} \\
        &= \Omega_k(d \sqrt{mw}).
    \end{align*}
    The second equality comes from plugging in the exact sizes of $X$ and $Y$. The final equality is implied by $w = O(m)$ and $d \geq 1$.
    Finally, note that $\log(\alpha/\beta) > \tau \sqrt{m\log(1/\eps)} \log m - w = \Omega(\sqrt{mw} \log m) > w + 2$ for large $m$.
    So, applying Lemma~\ref{l:pattern-matrix-method} and using the fact that $\log(1/\alpha) \leq w$, we get that the communication cost of the \QMAcc{} protocol must be $\Omega_k(Q/w) = d \sqrt{m/w}$.
\end{proof}

\begin{theorem}
    Let $n$ be a function of $m$ and let $\{f : \pmone^n \to \pmone\}_{m \in \N}$ be a family of functions. The following statements for all constants $k \geq 2$ and sufficiently large $m$.
    \begin{enumerate}
        \item For every $c > 0$, if $\deg_{1-1/m^c}(f) > d \geq 1$ then every \QMAcc{} protocol for $(\AND_m \circ f^m) \circ g_k^{nm}$ with witness length $w = O(m)$ must have communication cost $\Omega_k(d)$.
        \item If $\deg_{1/3}(f) > d \geq 1$ then every \QMAcc{} protocol for $(\AND_m \circ f^m) \circ g_k^{nm}$ with witness length at most $m - 4$ must have communication cost $\Omega_k(d)$.
    \end{enumerate}
\end{theorem}
\begin{proof}
    \textbf{(Part 1)}
    The proof is identical to that of Theorem~\ref{t:qma-cc-tradeoff-c-large}, except with $\eps = 2^{-m-2}$.
    Since the premises of
    Corollary~\ref{c:large-error-to-sbqp} are met with this setting of parameters, we get the necessary dual $\Phi$ via Theorem~\ref{t:sbqp-lb}.
    Following the proof of Theorem~\ref{t:qma-cc-tradeoff-c-large}, we get $\alpha > \eps = 2^{-m-2}$, $\beta \leq 2^{-\tau m \log m}$ ($\tau$ being a constant), and $Q = \Omega_k(dm)$.
    Consider any \QMAcc{} protocol with witness length $w \leq m$.
    That $w = O(m)$ implies $w \leq \log(\alpha/\beta) - 2$.
    Applying Lemma~\ref{l:pattern-matrix-method}, we get that the communication cost of the protocol must be at least $\Omega(Q/\log(1/\alpha)) = \Omega_k(d)$.

    \textbf{(Part 2)} We use the same proof again, except we obtain the dual $\Phi$ via Corollary~\ref{c:const-error-to-sbqp}. In this case, we get $\alpha > 2^{-m-2}$, $\beta \leq 2^{-2m}$, and $Q = \Omega_k(dm)$. We have $\log(\alpha/\beta) - 2 \geq m - 4 \geq w$, so the communication lower bound follows via Lemma~\ref{l:pattern-matrix-method}.
\end{proof}

By a straightforward application of these two theorems, our lower bounds on the \QMAdt{} complexity of the two-level AND-OR tree, surjectivity, and $k$-element distinctness (Theorem~\ref{t:qma-dt-lower-bounds}) yield matching \QMAcc{} lower bounds on those functions composed with an indexing gadget.